%% file: main.tex
\documentclass[runningheads,final]{llncs}
\usepackage[T1]{fontenc}
\usepackage[para]{footmisc}
\usepackage{graphicx}
\graphicspath{ {./images/} }
\usepackage[pdftex, colorlinks=true, hyperfootnotes=true, hyperindex=true, plainpages=false, pagebackref=false, pdfpagelabels=true, pdfstartview=FitH, linkcolor=blue, citecolor=blue, urlcolor=blue, bookmarks, bookmarksopen, bookmarksdepth=3]{hyperref}
\usepackage{color}

\input{addon/pkgloading}%
\input{addon/glossary}
\input{addon/glossary-BPM}

\input{addon/macros}
\input{addon/commands-for-inline-referenceable-objects-with-counter}
\input{addon/commands-for-pseudocode}

\usepackage{tikz}
\usetikzlibrary{tikzmark,decorations.pathreplacing}

\begin{document}
%
\title{A Secure, Confidential, and Verifiable \\ Decision Support System}
%
%
\author{Edoardo Marangone\inst{1}\orcidlink{0000-0002-0565-9168} \and
Eugenio Nerio Nemmi\inst{1}\orcidlink{0000-0001-6518-7863} \and
Daniele Friolo\inst{1}\orcidlink{0000-0003-0836-1735} \and \\
Giuseppe Ateniese\inst{2}\orcidlink{0000-0002-0848-878X} \and
Ingo Weber\inst{3}\orcidlink{0000-0002-4833-5921} \and
Claudio Di Ciccio\inst{4}\orcidlink{0000-0001-5570-0475}}
\authorrunning{E. Marangone et al.}
%
\institute{Sapienza University of Rome \email{\{lastname\}@di.uniroma1.it}
	\and
George Mason University \email{ateniese@gmu.edu}\\
	\and
Technical University of Munich \& Fraunhofer Gesellschaft \email{ingo.weber@tum.de}
	\and
Utrecht University \email{c.diciccio@uu.nl}
}
\maketitle              
\begin{abstract}
	\label{sec:abstract}%
\input{sections/abstract}
\end{abstract}

\section{Introduction}
\label{sec:intro}
\input{sections/intro2}
\section{Background}
\label{sec:background}
\input{sections/shortbackground}

\section{Problem illustration and requirements}
\label{sec:running-example}
\input{sections/example}

\section{Approach}
\label{sec:approach}
\input{sections/approach}
\subsection{Basic components and actors}
\label{subsec:approach:actorsAndComponents}
\input{sections/architecture}

\subsection{System initialization}
\label{sec:approach:system-init}
\input{sections/system_init}

\subsection{Execution}
\label{sec:approach:execution}
\input{sections/execution}

\subsection{Security Analysis}
\label{sec:formalAnalysis}
\input{sections/securityAnalysis}

\section{Implementation and evaluation}
\label{sec:evaluation}
\input{sections/evaluation_eu}


\section{Related Work}
\label{sec:sota}
\input{sections/sota}
\section{Conclusion, limitations and future work}
\label{sec:conclusion}
\input{sections/conclusion}

\bibliographystyle{splncs04}
\bibliography{bibliography}
\appendix
\section{Formal Analysis}\label{app:formal}
\input{sections/formalAppendix}
\end{document}

%% file: addon/pkgloading.tex
\usepackage[english]{babel}
\usepackage{subcaption}
\usepackage{caption}[singlelinecheck=false]
\usepackage{url}
\usepackage{colortbl}
\usepackage{amssymb}
\usepackage{stmaryrd}
\usepackage{amsmath}
\usepackage{mathrsfs}
\usepackage{amssymb}
\usepackage[left]{lineno}
\usepackage{xfrac}
\usepackage{nicefrac}
\usepackage[textsize=tiny,backgroundcolor=yellow!40,obeyFinal,textwidth=1.3cm]{todonotes}
\usepackage[nonumberlist,acronym]{glossaries}
\glsdisablehyper
\usepackage{comment}
\usepackage[capitalise,nameinlink]{cleveref}
\crefname{lstfloat}{listing}{listings}
\Crefname{lstfloat}{Listing}{Listings}
\crefname{section}{Sect.}{Sects.} 
\crefname{Section}{Section}{Sections}
\crefname{table}{Tab.}{Tabs.} 
\Crefname{table}{Table}{Tables}
\crefname{lstfloat}{List.}{Lists.} 
\Crefname{lstfloat}{Listing}{Listings}
\crefname{figure}{Fig.}{Figs.} 
\crefname{section}{Sect.}{Sects.} 
\crefname{appendix}{Appendix}{Appendices} 
\usepackage{tkz-base}
\usetikzlibrary{decorations.pathmorphing,trees,snakes,arrows,shapes,automata,petri}
\usepackage{paralist}
\usepackage{multicol}
\usepackage{booktabs}
\usepackage[inline]{enumitem}
\usepackage{multirow}
\usepackage{rotating}
\usepackage{etaremune}
\usepackage[ruled,linesnumbered,vlined,algo2e]{algorithm2e}
\SetNlSty{bfseries}{\color{black}}{}
\SetCommentSty{textrm} 
\SetArgSty{textrm} 
\SetFuncArgSty{textrm} 
\renewcommand{\SetProgSty}[1]{\renewcommand{\ProgSty}[1]{\textnormal{\csname#1\endcsname{##1}}\unskip}}
\SetProgSty{textrm} 
\usepackage{marginnote}
\usepackage{mathtools}
\usepackage{adjustbox}
\usepackage{ifthen}
\usepackage[normalem]{ulem}
\usepackage{lineno}
\usepackage{soul}
\usepackage{floatrow}
\usepackage{float}
\newfloat{lstfloat}{htbp}{lop}
\floatname{lstfloat}{Listing}

\usepackage{listings}
\lstdefinestyle{listing}{
	extendedchars=true,
	basicstyle=\fontsize{6pt}{6pt}\selectfont\ttfamily,
	showstringspaces=false,
	showspaces=false,
	numbers=left,
	numberstyle=\tiny,
	numbersep=0pt,
	tabsize=2,
	breaklines=true,
	showtabs=false,
	captionpos=b,
	lineskip = 2pt
}

\lstdefinestyle{table}{
	extendedchars=true,
	basicstyle=\fontsize{8pt}{8pt}\selectfont\ttfamily,
	showstringspaces=false,
	showspaces=false,
	numbersep=0pt,
	numbers=none,
	tabsize=2,
	breaklines=true,
	showtabs=false,
	captionpos=b,
	lineskip = 2pt
}
\definecolor{Orchid}{RGB}{150, 40, 120}

\crefname{lstlisting}{Listing}{Listings}
\Crefname{lstlisting}{Listing}{Listings}
\usepackage{lipsum}
\usepackage{makecell}
\usepackage{diagbox}
\usepackage[scientific-notation=false]{siunitx}
\usepackage{footmisc}
\usepackage{xspace}
\usepackage{ifdraft}
\usepackage{wrapfig}
\usepackage[percent]{overpic}
\usepackage{orcidlink}

%% file: addon/glossary.tex
\newacronym[\glslongpluralkey={Distributed Ledger Technologies}]{dlt}{DLT}{Distributed Ledger Technology}

\newacronym{ipfs}{IPFS}{InterPlanetary File System}
\newcommand{\IPFS}[0] {\Gls{ipfs}\xspace}

\newacronym{ipns}{IPNS}{InterPlanetary Name System}
\newcommand{\IPNS}[0] {\Gls{ipns}\xspace}

\newacronym{p2p}{P2P}{peer-to-peer}
\newcommand{\PtoP}[0] {\Gls{p2p}\xspace}

\newacronym{abe}{ABE}{Attribute-Based Encryption}

\newacronym{maabe}{MA-ABE}{Multi-Authority Attribute-Based Encryption}

\newacronym{cpabe}{CP-ABE}{Ciphertext-Policy Attribute-Based Encryption}

\newacronym{ssl}{SSL}{Secure Sockets Layer}

\newacronym{dht}{DHT}{Distributed Hash Table}

\newacronym{mpc}{MPC}{Multi-party Computation}

\newacronym{tee}{TEE}{Trusted Execution Environment}
\newcommand{\TEE}[0] {\Gls{tee}\xspace}

\newacronym{ecdh}{ECDH}{Elliptic-Curve Diffie-Hellman}
\newcommand{\ECDH}[0] {\Gls{ecdh}\xspace}

\newacronym{ecdsa}{ECDSA}{Elliptic Curve Digital Signature Algorithm}
\newcommand{\ECDSA}[0] {\Gls{ecdsa}\xspace}

\newacronym{sha}{SHA}{Secure Hash Algorithm}
\newcommand{\SHA}[0] {\Gls{sha}}

\newacronym{feel}{FEEL}{Friendly Enough Expression Language}
\newcommand{\FEEL}[0] {\Gls{feel}\xspace}

\newacronym{alfa}{ALFA}{Abbreviated Language for Authorization}
\newcommand{\ALFA}[0] {\Gls{alfa}\xspace}

\newacronym{cpu}{CPU}{Central Processing Unit}

\newacronym{aes-gcm}{AES-GCM}{Advanced Encryption Standard~-~Galois/Counter Mode}
\newcommand{\AESGCM}[0] {\Gls{aes-gcm}\xspace}

\newacronym{sgx}{SGX}{Software Guard Extensions}
\newcommand{\SGX}[0] {\Gls{sgx}\xspace}

\newacronym{tcp}{TCP}{Transmission Control Protocol}

\newacronym{tls}{TLS}{Transport Layer Security}
\newcommand{\TLS}[0] {\Gls{tls}\xspace}

\newacronym{dlib}{DeLib}{Decision Library}
\newcommand{\DLIB}[0] {\Gls{dlib}\xspace}

\newacronym{desobj}{DeSObj}{Decision Support Object}
\newcommand{\DESOBJ}[0] {\Gls{desobj}\xspace}

\newacronym{dectapp}{DeTApp}{Decision Trusted App}
\newcommand{\DECTAPP}[0] {\Gls{dectapp}\xspace}

\newacronym{ram}{RAM}{Random Access Memory}

\newglossaryentry{box}{ %
  name={Box},
  description={authority intialisation storage box},
  first={\glsentrydesc{box} (henceforth, \glsentrytext{Box} for short)},
  plural={Boxes},
  firstplural={\glsentrydesc{box}es (henceforth, \glsentryplural{box} for short)}
}

\newglossaryentry{rloc}{ %
	name={resource locator},
	description={a generic term to denote content-based links obtained via hashing (e.g., the IPFS link)},
}

\newacronym{cpmaabe}{CP-MA-ABE}{Ciphertext-Policy Multi-Authority Attribute-Based Encryption}

\newglossaryentry{dprovider}{ %
  name={Data Provider},
  description={provide data for decision making},
  plural={Data Providers},
}

\def\DOwner{\glsentrytitlecase{dprovider}{name}\xspace}

\def\Decider{Decider\xspace}

\def\DFile{Data File Processor\xspace}
\def\DNotary{Notarization System\xspace}

\newglossaryentry{po}{name={Policymaker}, description={defines decision rules, functions, and their callability}}
\newcommand{\PO}[0] {\Gls{po}\xspace}

\newacronym{ca}{CA}{Certification Authority}
\newcommand{\CA}[0] {\Gls{ca}\xspace}

\newacronym{ccu}{CCU}{Confidential Computing Unit}
\newcommand{\CCU}[0] {\Gls{ccu}\xspace}

\newacronym{sparta}{SPARTA}{Secure Platform for Automated decision Rules via Trusted Applications}
\newcommand{\SPARTA}[0] {\Gls{sparta}\xspace}

\newacronym{dk}{\textsl{dk}}{decryption key}
\newacronym{fdk}{\textsl{fdk}}{final decryption key}

\providecommand{\AdInt}{Admin Interface\xspace}
\providecommand{\UsInt}{User Interface\xspace}

\providecommand{\DataMan}{Data Manager\xspace}
\providecommand{\RequestMan}{Request Manager\xspace}

\newcommand{\Enc}{\mathsf{Enc}}
\newcommand{\Dec}{\mathsf{Dec}}
\newcommand{\Gen}{\mathsf{Gen}}
\newcommand{\keyspace}{\mathcal{K}}

\newcommand{\adv}{\mathcal{A}}
\newcommand{\NN}{\mathbb{N}}

\RequirePackage{hyperref}
\RequirePackage{xspace}
\newcommand{\tVAX}{\hyperref[sec:running-example]{VAX}\xspace}
\newcommand{\tMED}{\href{https://dmcommunity.org/challenge/challenge-feb-2021/}{MED}\xspace}
\newcommand{\tCARD}{\href{https://dmcommunity.org/challenge/challenge-oct-2021/}{CARD}\xspace}

\newacronym{pCERT}{CERT}{Certification}
\newcommand{\pCERTdesc}{\gls{pCERT}\label{step:sysboot:certificationauthority}}
\newcommand{\pCERT}{\hyperref[step:sysboot:certificationauthority]{\gls{pCERT}}\xspace}

\newacronym{pCCUINIT}{CCU-INIT}{CCU initialization}
\newcommand{\pCCUINITdesc}{\gls{pCCUINIT}\label{step:sysboot:ccu-init}}
\newcommand{\pCCUINIT}{\hyperref[step:sysboot:ccu-init]{\gls{pCCUINIT}}\xspace}

\newacronym{pSDEX}{SDEX}{Seed exchange}
\newcommand{\pSDEXdesc}{\gls{pSDEX}\label{step:sysboot:seedExchange}}
\newcommand{\pSDEX}{\hyperref[step:sysboot:seedExchange]{\gls{pSDEX}}\xspace}

\newacronym{pINIT}{CONF}{Communication channel configuration}
\newcommand{\pINITdesc}{\gls{pINIT}\label{step:sysboot:ccuInitialization}}
\newcommand{\pINIT}{\hyperref[step:sysboot:ccuInitialization]{\gls{pINIT}}\xspace}

\newacronym{pSPEC}{SPEC}{Specification of functions and policies}
\newcommand{\pSPECdesc}{\gls{pSPEC}\label{step:sysboot:policymaker}}

\newacronym{nRA}{RA}{A technical note on the remote attestation}
\newcommand{\nRAdesc}{\gls{nRA}\label{note:remote-attestation}}

\newacronym{nWP}{WS}{writing stage}
\newcommand{\nWPdesc}{\gls{nWP}\label{note:writing-phase}}
\newcommand{\nWP}{\hyperref[note:writing-phase]{\gls{nWP}}\xspace}

\newacronym{nWPi}{WS-i}{Client encrypted writing}
\newcommand{\nWPidesc}{\gls{nWPi}\label{step:writing-phase:client}}

\newacronym{nWPii}{WS-ii}{CCU encrypted writing}
\newcommand{\nWPiidesc}{\gls{nWPii}\label{step:writing-phase:tee}}
\newcommand{\nWPii}{\hyperref[step:writing-phase:tee]{\gls{nWPii}}\xspace}

\newacronym{DecMaking}{DECM}{Decision-making support}
\newcommand{\DecMakingdesc}{\gls{DecMaking}\label{step:dataflow:decision}}

%% file: addon/glossary-BPM.tex
\newacronym[\glslongpluralkey={Business Processes}]{bp}{BP}{Business Process}
\newacronym{bpi}{BPI}{Business Process Intelligence}
\newacronym{bpm}{BPM}{Business Process Management}
\newacronym{bpms}{BPMS}{Business Process Management System}
\newacronym{bpmn}{BPMN}{Business Process Model and Notation}
\newacronym{dmn}{DMN}{Decision Model and Notation}
\newacronym{dss}{DSS}{Decision Support System}
\newacronym{omg}{OMG}{Object Management Group}
\newacronym{cpn}{CPN}{colored Petri net}
\newacronym{kpi}{KPI}{Key Performance Indicator}
\newacronym{ocbc}{OCBC}{Object-centric Behavioral Constraints}
\newacronym{soa}{SOA}{Service-Oriented Architecture}
\newacronym{pn}{PN}{Petri net}
\newacronym{wf}{WF}{workflow}
\newacronym{wfms}{WfMS}{Workflow Management System}
\newacronym{xes}{XES}{eXtensible Event Stream}
\newacronym{yawl}{YAWL}{Yet Another Workflow Language}
\newglossaryentry{task}{%
	name={task},description={the non-divisible, elementary activity}}

\newglossaryentry{promod}{%
	name={process model},description={the model of a process}
}
\def\LogAlph {\ensuremath{\Sigma}}
\newglossaryentry{logalph}{
	name={log alphabet},description={the process alphabet, as reflected in a log},%
	symbol={\LogAlph}}
\def\Evt {\ensuremath{e}}
\newglossaryentry{evt}{
	name={event},description={a record of an instantaneous fact during the process enactment},%
	symbol={\Evt}}
\def\Trc { \ensuremath{\tau} }

\newglossaryentry{trace}{
	name={trace},description={a sequence of \glsplural{evt}},%
	symbol={\Trc}}
\def\EvtLog {\ensuremath{L}}

\newglossaryentry{evtlog}{
	name={event log},description={a collection of \glstext{evttrace}s},%
	symbol={\EvtLog}}

%% file: addon/macros.tex
\renewcommand\tabcolsep{5pt}
\newcolumntype{d}{>{\columncolor{gray!10}}c}
\newcolumntype{m}{>{\columncolor{gray!10}}l}
\newfloatcommand{capbtabbox}{table}[][\FBwidth]
\newenvironment{iiilist}%
{\begin{inparaenum}[\itshape(i)\upshape]}%
{\end{inparaenum}}

\newcommand{\mysubsubsection}[1]{\noindent \textbf{#1.}}

\newcommand\mypara[1]{\noindent \textit{#1.}}

\RequirePackage{xparse}
\NewDocumentEnvironment{AuthNote}{+o+o}{%
	\IfValueT{#2}{\marginnote{\scriptsize{#2}}}%
	\begin{scriptsize}
		\colorbox{gray}%
		{\color{white} Note\IfValueT{#1}{ (#1)}:}%
		\quad%
		\color{brown}
}{%
	\normalcolor
	\end{scriptsize}
}

\RequirePackage{pifont}

\RequirePackage{lipsum}
\newcommand{\LipsumGray}[1][]{{\color{gray}\ifthenelse{\equal{#1}{}}{\lipsum}{\lipsum[#1]}}}

\RequirePackage{siunitx}
\newcolumntype{D}[1]{S[
	table-omit-exponent,
	round-mode=places,
	round-integer-to-decimal,
	round-precision={#1}]} 

\providecommand{\ie}{{i.e.,}\xspace}
\providecommand{\eg}{{e.g.,}\xspace}

\newcommand{\SmallCode}[1]{\normalsize\texttt{#1}\normalsize}

%% file: addon/commands-for-pseudocode.tex
\DeclareDocumentCommand{\crefalgln}{ o m }{%
	\IfNoValueF{#1}{\cref{#1}, }
	{\hyperref[#2]{ln.~\ref{#2}}}
}%
\crefname{algocf}{Alg.}{Algs.}
\Crefname{algocf}{Algorithm}{Algorithms}
\crefname{aAlgocflLine}{line}{lines}
\Crefname{algocfline}{Line}{Lines}

\SetKw{Trigger}{trigger}
\SetKw{Yield}{yield}
\SetKw{Return}{return}
\SetKw{Except}{throw exception}
\SetKwProg{UponEv}{upon event}{ do}{}
\SetKwProg{Upon}{upon}{ do}{}
\SetKwInput{Implements}{Implements}
\SetKwInput{Uses}{Uses}
\SetKwProg{Fn}{Function}{ is}{}


%% file: sections/abstract.tex
Decision support systems are increasingly adopted to automate decision-making processes across industries, organizations, and governments. 
Decision support demands data privacy, integrity, and availability while ensuring customization, security, and verifiability of the decision process.
Existing solutions fail to guarantee those properties altogether. 
To overcome this limitation, we propose SPARTA, an approach based on Trusted Execution Environments~(TEEs) that automates decision processes.
To guarantee privacy, integrity, and availability, 
SPARTA employs efficient cryptographic techniques on notarized data with access mediated through user-defined access policies.
Our solution allows users to define decision rules, which are translated to certified software objects deployed within TEEs, thereby guaranteeing customization, verifiability, and security of the process.
With experiments run on public benchmarks and synthetic data, we show our approach is scalable and adds limited overhead compared to non-cryptographically secured solutions.

%% file: sections/intro2.tex
Organizations often need to make decisions collaboratively across trust boundaries while processing confidential data contributed by multiple independent parties. Such tasks are commonly supported by \gls{dss}, which are widely adopted across domains due to their ability to process large volumes of data, provide traceability, and enhance decision accuracy \cite{ARNOTT2008657,BARRACOSA2023114046,PEREIRA2022113795}. However, collaborative decision support creates a fundamental tension: the decision logic must be verifiable, and the process must be auditable. At the same time, the underlying data must remain confidential, including from the platform that executes the computation. 
A practical solution must satisfy four properties simultaneously: 
\begin{iiilist}
	\item confidentiality of input data and decision logic against all parties, including the computation platform; 
	\item verifiability of the decision process without exposing inputs; 
	\item user-defined decision rules with fine-grained access control; and 
	\item practical performance on realistic workloads. 
\end{iiilist}
No existing approach achieves all four properties. Blockchain-based solutions \cite{haarmann2018dmn} provide transparency but sacrifice data privacy. Homomorphic encryption and other privacy-preserving approaches \cite{Rahulamathavan} preserve confidentiality, but often incur prohibitive overhead, limit expressiveness, and may rely on trusted authorities for key management, thereby introducing single points of failure. Cloud-based platforms \cite{LIU2018825} offer flexible execution but require trusting a third party with plaintext data. Systems that support decision-process customization are typically designed as modules within centralized information systems, thereby guaranteeing neither privacy nor verifiability \cite{bazhenova2019bpmn}. \TEE and blockchain combinations, such as TEBDS \cite{TEBDS}, address some of these concerns, but still depend on centralized key management.

\begin{sloppypar}
We refer to this fourfold combination as the \textbf{confidential verifiable decision-support problem}, which introduces fundamental security challenges in multi-party settings.
As a solution to that, we present \SPARTA, an approach that leverages {\TEE}s to execute automated decision support in multi-party cooperation scenarios securely. \SPARTA executes data-driven decision functions based on custom decision logic inside attested trusted applications, stores encrypted data in a distributed file system, and uses a public blockchain as a notarization layer to guarantee traceability and integrity. Unlike prior approaches that rely on centralized key management, \SPARTA derives all keys within the \TEE from a shared seed exchanged via mutual attestation, thereby eliminating single points of failure. To express decision logic and access control, we adopt well-established standards, namely Decision Model and Notation (DMN) \cite{Dumas.etal/2018:FundamentalsofBPM} and \ALFA.\footref{foot:alfa} \SPARTA considers a threat model in which data providers are honest, the decision requester is potentially malicious, and the \TEE provides hardware-level isolation, excluding side channel attacks, which we discuss as a limitation.
We provide the following contributions:
\begin{compactenum}
	\item \textbf{Security architecture.} We design a protocol for confidential multi-party decision support that guarantees data privacy, integrity, and forward/backward secrecy by combining TEE-based execution with per-record key derivation and decentralized storage, without requiring a persistent trusted third party beyond initial certification.
	\item \textbf{Verifiable access-controlled computation.} We present a mechanism that translates user-defined decision logic (DMN) and attribute-based access policies (\ALFA) into attested, tamper-proof trusted applications, enabling fine-grained and verifiable authorization over confidential computations.
	\item \textbf{Security analysis and evaluation.} We formally define the security properties (correctness, data privacy, data integrity, forward/backward secrecy) of the system under a stated threat model, provide a security proof, and empirically demonstrate practical scalability on cross-domain benchmarks with an overhead of less than \qty{80}{\milli\second} compared to unprotected execution.
\end{compactenum}
\end{sloppypar}

%% file: sections/shortbackground.tex
Our approach is built upon the core pillars of confidential computing, and decision logic with access control.
\textbf{Confidential computing} is a paradigm wherein sensitive data undergoes processing within a protected processing unit \cite{ConfidentialComputingBook}.
Confidential computing underpins the~\textbf{\acrfullpl{tee}}~\cite{sabt2015trusted}. 
A \TEE  satisfies
\begin{iiilist}
	\item \textit{data confidentiality}: unauthorized entities cannot access data residing within the \TEE.
	\item \textit{data integrity}: unauthorized entities cannot modify, delete, or insert data inside the \TEE.
	\item \textit{code integrity}, unauthorized entities cannot manipulate, delete, or introduce alterations to the code \TEE. 
\end{iiilist}
Our work leverages Intel \SGX \cite{SGX1},  
facilitating the secure execution of code and data within \emph{enclaves}, namely hardware-protected environments. 
\SGX further enables \textit{remote attestation}, allowing external entities to authenticate a target enclave. 
One of the most well-known languages to model and formalize \textbf{decision logic} in business applications is \acrfull{dmn} \cite{dmn}, 
a standard by the \gls{omg}\footnote{\href{https://www.omg.org/spec/DMN}{\nolinkurl{omg.org/spec/DMN}}. Accessed: 2026-04-09
} 
for modeling and formalizing decision logic in multi-party processes and business applications.
A \gls{dmn}, is a decision table that captures critical knowledge and automates decision logic \cite{calvanese2016semantics}.
A \gls{dmn} table consists of columns representing the inputs and outputs of a decision, and rows denoting rules (the decision logic). Rules are expressed using \FEEL, a language defined within the \gls{dmn} standard. 
Each row has an identifier, a condition for each input column, and one specific value for each output column. 
A rule is considered applicable when the input values satisfy all associated conditions; in that case, the corresponding output values are produced \cite{calvanese2016semantics}. 
To express \textbf{authorization policies} over the usage of resources, the
\ALFA\footnote{\href{https://alfa.guide/alfa-authorization-language/}{\nolinkurl{alfa.guide/alfa-authorization-language/}}. Accessed: 2026-04-09
\label{foot:alfa}} language provides a concise, human-readable format. Built on top of the eXtensible Access Control Markup Language~(XACML),\footnote{\href{https://docs.oasis-open.org/xacml/3.0/xacml-3.0-core-spec-os-en.html}{\nolinkurl{docs.oasis-open.org/xacml}}. Accessed: 2026-04-09
} \ALFA offers a simplified syntax for defining fine-grained, attribute-based access control~(ABAC) policies based on users, actions, resources, and context. 

In addition, our architecture employs two additional buttresses for the notarization and storage of operations and data.
\textbf{Blockchains} are Distributed Ledger Technologies~(DLT), wherein transactions are sorted, organized into blocks, and interconnected to form a chain possessing resistance against tampering, achieved through cryptographic methods. 
Public blockchains like Ethereum 
and Algorand 
permit the deployment and execution of programs named smart contracts \cite{IntroducingEthereumandSolidity}. 
To mitigate the expenses associated with the invocation of smart contracts, external \PtoP systems are frequently employed for storing large volumes of data \cite{ArchitectureforBlockchainApplications}, such as \textbf{\IPFS} \cite{ipfs},\footnote{\label{foot:ipfs} \href{https://ipfs.tech/}{\nolinkurl{ipfs.tech}}. Accessed 2026-03-19.} a distributed system designed for file storage and access.

%% file: sections/example.tex
To illustrate the problem we address, we introduce a real-world scenario from the healthcare domain: a vaccine distribution campaign inspired by \cite{vaccinedistribution,noh2021group} and linked to a research project funded by the European NextGenerationEU.\footnote{SERICS (PE00000014), funded by the Italian Ministry of University and Research.} 
We use this scenario as a running example and motivating case study throughout the paper.
In our scenario, a central medical hub (hereafter, medical hub) receives vaccine supplies and manages their dispatch within a federal State; patients receive vaccine injections, vaccination centers administer them, and specialized carriers deliver the vaccines on-site.

In our scenario, we distinguish the following user \emph{roles}: medical hub, patient, vaccination center, and carrier. A role is one of the \emph{attributes} that characterize a user.
\begin{table}[b]
	\caption{Examples of users, groups, and attributes}%
	\label{tab:attributes}%
	\centering
	\resizebox{\columnwidth}{!}{
		\input{tables/attributes}
	}
\end{table}
In the first stage of our running example, users belonging to the \emph{\DOwner} group contribute data to the system; we refer to this stage as \emph{data provision}. \Cref{tab:attributes} reports examples of user roles, groups, and attributes, while \Cref{tab:data} shows part of the data that these users provide to the system. For instance, the National vaccine hub, which serves as a medical hub in Italy, reports vaccine stock levels and storage temperatures. Andre Smith, a patient from Sardinia, submits vaccination requests together with personal information, such as pre-existing conditions and family medical history. Ayala PLC, a vaccination center in Tuscany, reports storage capacity and vaccination progress. PrimeWay, a carrier, registers service-related information such as the number of trucks and refrigeration capacity.
%
\begin{table}[tb]
	\caption{An example of the data fields users provide to the system}%
	\label{tab:data}%
	\centering
	\input{tables/inputs}
\end{table}

In the subsequent \emph{data usage} stage, authorized users can access the data provided earlier, either at the level of single entities or in aggregated form. In our example scenario, patients can retrieve their declarations at any time, and carriers can obtain information about on-time delivery rates. Users in the \emph{\Decider} group can also request automated decision support. The medical hub, e.g., belongs to this group and can request decision support for vaccine supply distribution.
Decisions related to vaccine distribution include
\begin{iiilist}
	\item the priority to be assigned to patients for vaccine administration (which we use as a running example throughout this paper),
	\item the restocking needs of each vaccination center, and
	\item the carrier that should handle each delivery.
\end{iiilist}

In our example, several critical aspects need to be considered. 
First, all data is confidential and must be protected from unauthorized access.
For example, revealing a single patient's metadata constitutes a privacy breach. 
Second, data integrity and availability are critical.
Any unauthorized modification could alter decisions or falsify patient records, while system failures or data unavailability could disrupt the decision-making process.
Ensuring privacy, integrity, and availability while enabling customizable, secure, and verifiable automated decision-making is an open challenge that motivates our research.

\Cref{tab:requirements} presents the requirements identified from the motivating use case that underpins our approach.
\begin{table}[tb]
	\caption{Requirements and corresponding actions}%
	\label{tab:requirements}%
	\centering
	\resizebox{\textwidth}{!}{
		\input{tables/requirements}
	}
\end{table}
Ensuring that decision-critical information is consistently available, tamper-evident, and auditable is essential to both the efficiency of healthcare processes and trust in the system.
The information source must be immutable and accessible, with traceability ensured to facilitate subsequent validations and verifications. We formalize these needs in \cref{req:availability}. 
Since the system manages sensitive and confidential healthcare data, its security is crucial to ensuring accurate decisions and preventing unauthorized access, data leaks, malicious attacks, and system unavailability. For example, disclosing patients' metadata or their vaccination prioritization would violate their privacy. We enclose these properties in \cref{req:security}. 
Finally, our approach must be interoperable with existing data storage, access, and aggregation solutions while enabling secure, automated decision-making for confidential data.
We consider these aspects in \cref{req:interoperability}.
Our solution aims to meet these requirements.

%% file: tables/attributes.tex

\begin{tabular}{lll}
	\toprule
	\textbf{User} & \textbf{Group} & \textbf{Attributes} \\
	\midrule
	
	Andre Smith & Data Provider &
	\SmallCode{Role="Patient"}, \SmallCode{Region="Sardinia"} \\
	
	
	National vaccine hub & Data Provider, Decider &
	\SmallCode{Role="CentralMedicalHub"}, \SmallCode{Country="Italy"} \\
	
	
	Ayala PLC & Data Provider &
	\SmallCode{Role="VaccinationCenter"}, \SmallCode{Region="Tuscany"} \\
	
	PrimeWay & Data Provider &
	\SmallCode{Role="Carrier"}, \SmallCode{Area="North"} \\
	
	
	\bottomrule
\end{tabular}

%% file: tables/inputs.tex
\resizebox{\textwidth}{!}{%
	\begin{tabular}{llllll}
		\toprule
		\multicolumn{3}{c}{\textbf{Central medical hub}} &
		\multicolumn{3}{c}{\textbf{Vaccination center}} \\
		\cmidrule(r){1-3} \cmidrule(r){4-6}
		\multicolumn{1}{l}{Field} & \multicolumn{1}{l}{Type} & \multicolumn{1}{l}{Value} &
		\multicolumn{1}{l}{Field} & \multicolumn{1}{l}{Type} & \multicolumn{1}{l}{Value} \\
		\cmidrule(r){1-3} \cmidrule(r){4-6}
		Stock ID             & Integer   & \texttt{10523}                          & Hub Name                    & String  & \texttt{Ayala PLC}                     \\
		Quantity Available   & Integer   & \texttt{2400}                           & Max Storage Capacity        & Integer & \texttt{2000}                          \\
		Storage Temperature  & Float     & \texttt{-20.5}~{[°C]}     & Vaccination Progress        & Integer & \texttt{686}~{[per day]} \\
		\bottomrule
	\end{tabular}%
}

\bigskip

\resizebox{\textwidth}{!}{%
	\begin{tabular}{llllll}
		\toprule
		\multicolumn{3}{c}{\textbf{Patient}} &
		\multicolumn{3}{c}{\textbf{Carrier}} \\
		\cmidrule(r){1-3} \cmidrule(r){4-6}
		\multicolumn{1}{l}{Field} & \multicolumn{1}{l}{Type} & \multicolumn{1}{l}{Value} &
		\multicolumn{1}{l}{Field} & \multicolumn{1}{l}{Type} & \multicolumn{1}{l}{Value} \\
		\cmidrule(r){1-3} \cmidrule(r){4-6}
		Full Name               & String    & \texttt{Andre Smith} & Carrier Name                   & String  & \texttt{PrimeWay}            \\
		Pre-existing Conditions & String    & \texttt{Asthma}         & Number of Trucks               & Integer & \texttt{13}                          \\
		Family Medical History  & String    & \texttt{Heart Disease}  & Refrigeration Capacity         & Boolean & \texttt{True}                        \\
		\bottomrule
	\end{tabular}%
}

%% file: tables/requirements.tex
\begin{tabular}{p{0.3cm} p{5cm} p{6cm}}
	\toprule
	& \textbf{Requirement} & \textbf{Approach} \\
	\midrule
	{\Req\label{req:availability}} & The stored \textbf{information} shall be tamper-proof, traceable, confidential, permanent, and available & We store the encrypted data in a tamper-proof distributed file storage, and we use a notarization system to keep track of information \\
	{\Req\label{req:security}} & \textbf{Decision support} shall be user-definable, trusted, reproducible, verifiable, and confidential & We handle data within trusted hardware, encrypting them before being stored outside of it; data access and processing is established through customizable policies, turned into verifiable apps running in trusted hardware \\
	{\Req\label{req:interoperability}} & The decision support \textbf{system} must be secure, interoperable, available, replicable, and scalable & We employ standard languages to describe authorization policies (expressing what each user can do and what resources they may access) and decision logic (the rules to follow in making decisions) \\
	\bottomrule
\end{tabular}

%% file: sections/approach.tex
In this section, we present our approach, named \acrfull{sparta}.
In our description, we motivate design choices and operations with the previously discussed requirements (see \cref{sec:running-example,tab:requirements} above).
To this end, we first provide an overview of its architecture (\cref{subsec:approach:actorsAndComponents}). Then, we describe how we initialize the system (\cref{sec:approach:system-init}) and 
thereby execute secure decision processes over confidential data (\cref{sec:approach:execution}). Finally, we outline a security analysis of our approach (\cref{sec:formalAnalysis}).

%% file: sections/architecture.tex
\begin{figure}[tb]
	\centering
	\begin{overpic}[width=\columnwidth]{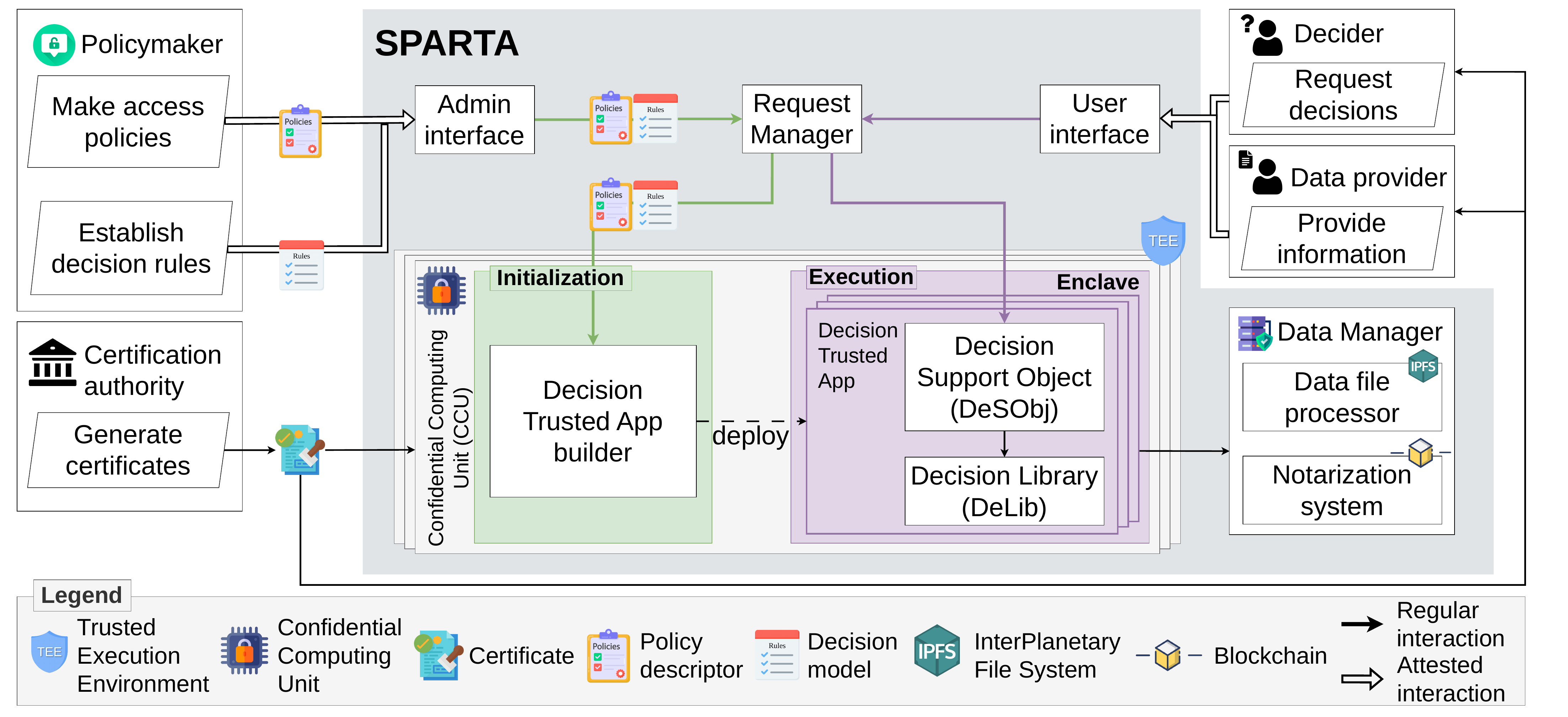}
		\put(15.7,39.5){\footnotesize\cref{lst:alfa-descriptor}}
		\put(15.7,31){\footnotesize\cref{tab:dmn}}
	\end{overpic}
	\caption[An overview of the approach]{An overview of the \SPARTA approach with the software components, user roles, and main information flow
	}
	\label{fig:general-architecture}
\end{figure}

\Cref{fig:general-architecture} provides a bird's eye view of the \SPARTA approach.
%
%
%
    The \textbf{\CCU} is the core component of \SPARTA. It is deployed inside a \TEE and can be replicated to improve data availability and system robustness, as required by \cref{req:availability,req:security}. 
	Dedicated {\DECTAPP}s execute a decision process each within their own enclave inside the \CCU. A \DECTAPP builder encodes and deploys them based on two inputs:
    \begin{iiilist}
    \item a \textit{decision model} (see \cref{tab:dmn}), which defines the functions exposed by the \DECTAPP, and
    \item a \textit{policy descriptor} (see \cref{lst:alfa-descriptor}), which specifies who can invoke those functions.
    \end{iiilist}
	More specifically, those inputs customize what we call a \DESOBJ. Internally, the \DECTAPP links the \DESOBJ with a \DLIB, namely a shared library we developed for maintainability and encapsulation purposes. It provides common utilities for {\DESOBJ}s, such as data retrieval and decryption, policy checking, and certificate validation.
    The \textbf{\RequestMan} acts as a fa\c{c}ade to the \CCU. It can be abstracted as a queue that handles both requests to and responses from the \CCU. Its purpose is to save them and preserve their order while the \CCU is handling other requests, thereby avoiding concurrency issues. The \RequestMan operates outside the \TEE, allowing the \DECTAPP to perform security-critical operations inside the enclave while avoiding overuse of protected memory. Consequently, all data exchanged through the \RequestMan must be encrypted, as we explain in \cref{sec:approach:execution}. The \RequestMan serves two entry points: one for system-initialization operations (\AdInt), such as certificate management and the submission of decision models and policy descriptors, and one for runtime interactions (\UsInt), through which {\DOwner}s and the \Decider communicate with the \CCU.
    The \textbf{\DataMan} is an external component providing two architectural buttresses:
	\begin{inparaenum}[\itshape(i)\upshape]
		\item The \DFile, which stores the content of the exchanged data, and
		\item the \DNotary, which saves the resource locators pointing to that data.
	\end{inparaenum}
    We realize these components through two established technologies to meet \cref{req:availability}.
	We employ \IPNS over \IPFS for the \DFile since it provides permanence, immutability, and content addressing for stored files. 
	We use a public blockchain as the \DNotary to store \IPNS names in a tamper-proof, available, and traceable way.
	Other solutions like STORJ\footnote{\href{https://www.storj.io/}{\nolinkurl{storj.io}}. Accessed: 2026-03-19.}
	and distributed relational databases~\cite{distributed-relational-database} may be used in scenarios with less restrictive security requirements.

We identify four main actors, each responsible for a distinct functionality.
\begin{inparaitem}[]
	\item The \textbf{\acrfull{ca}} is the trusted authority that issues digital certificates to users, certifying their \emph{attributes}. 
	In our example, the role of the \CA should be fulfilled by a globally recognized Certification Authority~(\eg \textit{DigiCert}, \textit{Entrust}, or \textit{GlobalSign}),\footnote{DigiCert: \href{https://www.digicert.com}{\nolinkurl{digicert.com}}; Entrust: \href{https://www.entrust.com}{\nolinkurl{entrust.com}}; GlobalSign: \href{https://www.globalsign.com/en}{\nolinkurl{globalsign.com}}. Accessed: 2026-04-07.} akin to what is offered by the National Health Service in the United Kingdom \cite{DBLP:journals/istr/Dalton03}. 
	\item The \textbf{\PO} defines the decision model and the policy descriptor needed for the deployment of the \DECTAPP. Considering the example in \cref{sec:running-example}, the \PO role can be fulfilled by a regulatory body such as the Ministry of Health, or by the consortium directly involved in the decision process, as in \cite{iacr/CerulliCNPS23}, upon reaching consensus.
	\item The \textbf{\DOwner} provides the input data and contextual information used for decision-making. In our running example, the \DOwner group comprises the roles medical hub, patient, vaccination center, and carrier.
	Finally, 
	\item the \textbf{\Decider} receives decision support based on the available data. If permitted by access policies, the \Decider may also retrieve raw data and invoke aggregation functions. In our scenario, the medical hub role acts as a \Decider.
\end{inparaitem}

%% file: sections/system_init.tex

The system initialization phase goes through three consecutive steps:
\begin{inparaenum}[\itshape(i)\upshape]
\item \acrfull{pCERT}, 
\item \acrfull{pCCUINIT}; 
\item \acrfull{pSPEC}. 
\end{inparaenum}
We detail them in the following.


\mysubsubsection{\pCERTdesc}
In this step, the \CA certifies each actor's attributes by issuing a certificate. For our purposes, we consider X.509 certificates. Each certificate contains the actor's metadata, including attributes and public key. For the latter, we leverage \ECDSA \cite{ECDSA}. Before any interaction, each participant must verify the counterpart's certificate~(see \cref{req:interoperability}). The \CCU must likewise hold a valid certificate that allows other actors to confirm they are interacting with the intended \CCU~(see \cref{req:security}). To sign messages, actors use the \ECDSA private key corresponding to the public key attested to in their certificate; otherwise, the signature is invalid.

\mypara{\nRAdesc} 
Before making any request, a client (a \DOwner, \Decider, or another \CCU) performs remote attestation to establish trust in the \CCU inside the \TEE~(see \cref{req:security}). In this phase, we follow the RATS RFC standard \cite{birkholz2023rfc}, the foundation for attestation schemes such as Intel EPID and AMD SEV-SNP. Remote attestation allows the client to verify that they are sending data to a Trusted Application within a \TEE and to authenticate the \DESOBJ as a legitimate \CCU entity certified by the \CA (see above).
We slightly modify the standard procedure. The client:
\begin{iiilist}
	\item\label{ra:ca-cert}retrieves and verifies the \CCU's \CA-issued certificate using the \CA's \ECDSA public key, confirming the \CCU is the intended party;
	\item\label{ra:tls}retrieves and validates the \CCU's \TLS certificate (including its \ECDSA public key) to establish a secure connection; and
	\item\label{ra:report}obtains and verifies the \CCU's attestation report over the attested \TLS channel. The report contains the \emph{measurement} of the \DECTAPP, i.e., the hash of its code and data, signed with the \CCU's hardware attestation key. The client validates it by authenticating the \CCU, decrypting the report, comparing the reported measurement with the expected reference value, and verifying that the report is consistent with the server's \TLS certificate.
\end{iiilist}
		
\mysubsubsection{\pCCUINITdesc} 
This step sets up the replicated {\CCU}s and the secure communication with them. It encompasses two phases: \textbf{\gls{pSDEX}} and \textbf{\gls{pINIT}}.
\textbf{\pSDEXdesc} addresses the following two issues. Relying on a single \CCU introduces significant vulnerabilities that could compromise \cref{req:availability,req:security}. Should the component fail, experience downtime, or be subject to malicious attacks, the entire execution process may halt indefinitely, and recovery may become impossible. Also, if that \CCU alone held the cryptographic keys, all associated data would be permanently lost.
To address these issues, \SPARTA's {\CCU}s share a common seed, enabling all replicas to deterministically derive the necessary cryptographic keys (meeting \cref{req:security}).
To describe how \pSDEX unfolds, we shall use an example with two {\CCU}s: \CCU{1} and \CCU{2}. \CCU{1} generates and sends the seed. It proceeds as follows:
\begin{inparaenum}[\itshape(i)\upshape]
	\item it performs the remote attestation of \CCU{2};
	\item if remote attestation (\nRAdesc) succeeds, it generates a secure random string~(the seed), seals it with a sealing key known only to itself~(see \cref{sec:background}), and sends the seed to \CCU{2} through the secure channel established during \nRAdesc.
\end{inparaenum}
After receiving the seed from \CCU{1}, \CCU{2} performs the same operations: it remotely attests \CCU{1} and, if successful, seals the seed with a sealing key known only to itself.
Once exchanged, the seed can be shared with an arbitrary number of additional {\CCU}s. We describe how it is used to derive encryption keys in \cref{sec:approach:execution}.
The subsequent phase (\textbf{\pINITdesc}) generates a \TLS certificate. This certificate serves two purposes: It enables the establishment of a secure communication channel with the \CCU during remote attestation, and provides the client with the shared secret needed to encrypt data before sending it through the \RequestMan. Since the \RequestMan operates outside the \TEE (see \cref{subsec:approach:actorsAndComponents}), all data exchanged with the system must be encrypted.
To do so, we slightly modify the \TLS certificate that each \CCU generates, allowing it to embed a public key deterministically generated from the shared seed~(see \pSDEX). Since all {\CCU}s share the same seed, they produce identical key pairs. We leverage \ECDH \cite{ECDH}~(X25519 curve \cite{Curve25519}) for the key pair generation. The \CCU signs the \TLS certificate with its \ECDSA private key, which is \CA-certified~(see \pCERT), allowing any participant to verify it. 
At the end of the {\pINIT} phase (and hence of the {\pCCUINIT} step), the \CCU is ready to receive requests and provide responses through the \RequestMan.
    

	
\begin{sloppypar}
\mysubsubsection{\pSPECdesc}
With this step, the \PO specifies the decision functions the system should realize in a decision model, and their access policies in a policy descriptor. To meet \cref{req:interoperability}, the \PO expresses these using standard, well-established, and customizable languages: DMN with \FEEL and \ALFA, respectively (see \cref{sec:background}). Their details transcend the scope of this paper, 
and illustrate their use by example here.
\begin{table}[b]
	\caption{Decision model for the patient prioritization of \cref{sec:running-example}}
	\label{tab:dmn}
	\resizebox{\columnwidth}{!}{
		\input{tables/dmn}
	}
\end{table}
\Cref{tab:dmn} shows a DMN table for prioritizing vaccine requests by patients. The \emph{source data} input columns refer to patient data provided in the \emph{data provision} stage in \cref{sec:running-example}. DMN also allows the specification of \FEEL expressions that filter and then aggregate data. For example, the \SmallCode{StockCoverageDays} input column is a \FEEL expression that calculates the number of days the currently available vaccine stock can sustain, based on vaccination center activity and vaccination progress. 
The corresponding \FEEL expression 
is:
\begin{lstlisting}[escapechar=§,style=listing,numbers=none,linerange={27-31}]
{ totalAvailableVaccines: sum(medicalHub.QuantityAvailable), 
  activeCenters: vaccinationCenters[VaccinationProgress > 0],
  totalDailyVaccinations: sum(activeCenters.VaccinationProgress),
  stockCoverageDays: totalAvailableVaccines / totalDailyVaccinations }.stockCoverageDays
\end{lstlisting}
The \emph{decision} output column shows the corresponding decision.
\begin{lstfloat}[tb]
	\caption{An excerpt of a policy descriptor file in ALFA}
	\label{lst:alfa-descriptor}
\begin{lstlisting}[escapechar=§,style=listing]
namespace VaccineDispatch {
	condition isMedicalHub Role=="CentralMedicalHub" and Country=="Italy"
	policy PatientPriorityWAggr {
		target clause Action=="PatientPriorityWAggr(Patient,stockCoverageDays)"
		rule accessDecisionWAggr { condition isMedicalHub } } }
\end{lstlisting}
\end{lstfloat}
\Cref{lst:alfa-descriptor} shows an excerpt of the \ALFA descriptor file.%
\footnote{Full descriptor: \href{https://github.com/apwbs/SPARTA/blob/main/src/alfaparser/policy.alfa}{\nolinkurl{github.com/apwbs/SPARTA/blob/main/src/alfaparser/policy.alfa}}}
In our running example, only a medical hub located in Italy may invoke the \SmallCode{PatientPriorityWAggr} 
function, which determines the priority level assigned to a patient for vaccine administration based on patient data and the aggregated \SmallCode{StockCoverageDays} input. The \SmallCode{accessDecisionWAggr} rule enforces this restriction.
Finally, the \DECTAPP builder compiles the ALFA policy and the DMN table into a \DESOBJ, which is deployed as a \DECTAPP within the \TEE enclave. Thereafter, the \Decider can invoke the decision function. As we describe next, the \DESOBJ makes a decision within the enclave, so the invoker reads the result but not the input data.
\end{sloppypar}

%% file: tables/dmn.tex
{%
	\begin{tabular}{ccccccc|c}
		\toprule
		\multicolumn{1}{l}{{\textsf{When}}} &
		\multicolumn{1}{l}{{\textsf{and}}} &
		\multicolumn{1}{l}{{\textsf{and}}} &
		\multicolumn{1}{l}{{\textsf{and}}} &
		\multicolumn{1}{l}{{\textsf{and}}} &
		\multicolumn{1}{l}{{\textsf{and}}} &
		\multicolumn{1}{l|}{{\textsf{and}}} &
		\multicolumn{1}{l}{{\textsf{then}}} \\
		
		\multirow{2}{*}{\makecell{\textbf{Age}}} &
		\multirow{2}{*}{\makecell{\textbf{Pre}\\\textbf{Existing}\\\textbf{Conditions}}} &
		\multirow{2}{*}{\makecell{\textbf{Current}\\\textbf{Medications}}} &
		\multirow{2}{*}{\makecell{\textbf{Previous}\\\textbf{Vaccinations}}} &
		\multirow{2}{*}{\makecell{\textbf{Family}\\\textbf{Medical}\\\textbf{History}}} &
		\multirow{2}{*}{\makecell{\textbf{Consent}\\\textbf{Form}\\\textbf{Signed}}} &
		\multirow{2}{*}{\makecell{\textbf{Stock}\\\textbf{Coverage}\\\textbf{Days}}} &
		\multirow{2}{*}{\makecell{\textbf{Patient}\\\textbf{Priority}\\\textbf{WAggr}}} \\
		& & & & & & & \\
		
		\multicolumn{1}{r}{\textsf{\tiny{:number}}} &
		\multicolumn{1}{r}{\textsf{\tiny{:string}}} &
		\multicolumn{1}{r}{\textsf{\tiny{:string}}} &
		\multicolumn{1}{r}{\textsf{\tiny{:string}}} &
		\multicolumn{1}{r}{\textsf{\tiny{:string}}} &
		\multicolumn{1}{r}{\textsf{\tiny{:boolean}}} &
		\multicolumn{1}{r|}{\textsf{\tiny{:number}}} &
		\multicolumn{1}{r}{\textsf{\tiny{:string}}} \\
		\midrule
		
		\multirow{2}{*}{$\geq$60} &
		\multirow{2}{*}{\makecell{``Asthma'',\\``Diabetes''}} &
		\multirow{2}{*}{\makecell{``Metformin'',\\``Albuterol''}} &
		\multirow{2}{*}{\makecell{``COVID-19'',\\``Influenza''}} &
		\multirow{2}{*}{\makecell{``Diabetes'',\\``Heart Disease''}} &
		\multirow{2}{*}{true} &
		\multirow{2}{*}{$\leq$3} &
		\multirow{2}{*}{``High''} \\
		& & & & & & & \\
		
		{[}18..60{[} & ``Asthma'' & ``Metformin'' & ``Influenza'' & - & true & {]}3..7{]} & ``Medium'' \\
		
		\textless{}18 & - & ``Lisinopril'' & ``COVID-19'' & - & true & \textgreater{}7 & ``Low'' \\
		
		\tikzmarknode{startbrace}{-} & - & - & - & - & \tikzmarknode{endbrace}{false} &
		\tikzmarknode{aggstart}{-} & \tikzmarknode{decstart}{\tikzmarknode{decend}{``Ineligible''}} \\
		\bottomrule
	\end{tabular}%
	\begin{tikzpicture}[overlay,remember picture]
		\draw[decorate, decoration={brace, amplitude=8pt, mirror}]
		([xshift=-15pt,yshift=-8pt]startbrace.south west) --
		([xshift=12pt,yshift=-8pt]endbrace.south east)
		node[midway, below=8pt] {\footnotesize\textit{source data}};
		
		\draw[decorate, decoration={brace, amplitude=8pt, mirror}]
		([xshift=-10pt,yshift=-8pt]aggstart.south west) --
		([xshift=11pt,yshift=-8pt]aggstart.south east)
		node[midway, below=8pt] {\footnotesize\textit{FEEL expr.}};
		
		\draw[decorate, decoration={brace, amplitude=8pt, mirror}]
		([xshift=-4pt,yshift=-7pt]decstart.south west) --
		([xshift=5pt,yshift=-7pt]decend.south east)
		node[midway, below=8pt] {\footnotesize\textit{decision}};
	\end{tikzpicture}%
}

%% file: sections/execution.tex
The lilac box in \cref{fig:general-architecture} encloses the components involved in the execution phase of \SPARTA. It is comprised of 
a \textbf{\nWPdesc}, wherein data is provided in an encrypted form, and the \textbf{decision-making stage (DECM)}, in which automated decisions are made.
These stages correspond, respectively, to the \emph{data provision} and \emph{data usage} stages of the running example in \cref{sec:running-example}.


\mysubsubsection{\nWPidesc}
In this step, the \DOwner encrypts the data before transmitting it through the \RequestMan, which operates outside the trusted environment. 
Since multiple {\CCU}s are possible, the \DOwner needs to encrypt the data with a key known to all {\CCU}s. To achieve this, the \DOwner retrieves a \CCU's \TLS certificate (the one generated during the \pINIT phase), and if the remote attestation completes successfully, derives the shared secret from it and generates the symmetric key required to encrypt the data, thereby meeting \cref{req:availability}.
Upon receiving the request, the \CCU verifies the \DOwner's certificate and signature, derives the same shared secret, reconstructs the symmetric key, and decrypts the payload, thus meeting \cref{req:security}. 
This mechanism relies on an \ECDH key exchange over the X25519 curve \cite{ECDH,Curve25519}, \SHA-256 \cite{SHA2}, 
and \AESGCM \cite{AES-GCM2}. 

\mysubsubsection{\nWPiidesc}
This step begins after the \CCU decrypts the \DOwner's input data. The \CCU must encrypt the data before storing it outside the enclave, as it would otherwise be accessible in plaintext. To guarantee this, the \CCU generates a distinct symmetric key for each record from the shared seed established during \pSDEX and a newly generated random number. Thereafter, the \CCU encrypts the data with this key using \AESGCM to meet \cref{req:security}. It stores the ciphertext on \IPFS 
and records the required \IPNS names in a smart contract on the public blockchain in compliance with \cref{req:availability}.
The random number is stored in plaintext alongside the encrypted data on \IPFS; this does not compromise security, since the encryption key depends on both the random number and the private seed. Consequently, an attacker cannot derive valid keys without access to the seed. Indeed, without a valid \CA certificate, an attacker cannot participate in \pSDEX,
while the seed itself remains protected by a sealing key. Finally, we remark that if a \CCU becomes unavailable, any other \CCU that participated in \pSDEX can take over, as it possesses the same seed.

\begin{sloppypar}
\mysubsubsection{\DecMakingdesc}
In this step, the \Decider requests decision support. To do so, it first performs remote attestation against the \CCU. Then it invokes a decision function exposed by the \DESOBJ deployed within it (e.g., \SmallCode{PatientPriorityWAggr}), thereby meeting \cref{req:security}. The request includes the \Decider's certificate issued by the \CA, the name of the invoked function, and the \IPNS key name identifying the relevant data. After parsing the request, the \CCU verifies the \Decider's certificate and checks whether the certified attributes satisfy the access policy associated with the invoked function, as required by \cref{req:interoperability}. In our running example, this policy requires the invoker to be a medical hub located in Italy. If the policy is satisfied, the \CCU retrieves the ciphertext from \IPFS, recovers the symmetric key, decrypts the data, evaluates the required aggregation functions, and executes the decision function within the enclave, thereby preserving confidentiality and integrity in compliance with \cref{req:security}. In our example, this amounts to computing the \SmallCode{StockCoverageDays} aggregation before evaluating \SmallCode{PatientPriorityWAggr}. Finally, the \CCU returns the result to the \Decider through the \RequestMan.
\end{sloppypar}

%% file: sections/securityAnalysis.tex
To examine the security of our system, we first provide a threat model and the desired properties we aim to achieve.

\mysubsubsection{Threat Model}\label{formal-analysis:thread-model} 
\DOwner and \CA are honest, the \Decider 
is dishonest 
and secure authenticated channels can be established between authorized parties (\eg CA, Decider) and the \CCU.
We model the \CCU as a fully trusted third-party computing decisions through functions $\{f_j\}$ using a secret seed $s$.

\mysubsubsection{Desired Properties}\label{formal-analysis:desired-properties}
We informally specify the following required properties:

\mypara{Correctness}
For any function $f$ and any set of messages $M$ such that $f(M)=y$, if $M$ is provided to the~\CCU for function $f$ by the Data Provider obtaining ciphertexts $C$, and subsequently, ciphertexts $C$ are provided by the Decider back to the \CCU obtaining output $y'$, then $y'=f(M)=y$.

\mypara{Data Privacy}
An adversary cannot distinguish between the encryption of two messages giving the same output for an adversarially chosen function $f$.

\mypara{Data Integrity}
An adversary cannot tamper with encrypted messages.

\mypara{Forward/Backward Secrecy/Integrity}
If all ciphertext keys at time $t$ are leaked, ciphertexts encrypted at any different time $t'$ remain secure, and their underlying messages cannot be decrypted or forged.

\begin{theorem}[Informal]
    Assuming that the authenticated encryption scheme is secure, the hash function is modeled as a random oracle, and the TEE is secure, our system satisfies Correctness, Data Privacy, Data Integrity, and Forward/Backward Secrecy and Integrity.
\end{theorem}

The theorem follows from the security of the \TEE and the indistinguishability and untamperability properties of the underlying authenticated encryption scheme. For the complete formal analysis and proof, see \cref{app:formal}.

%% file: sections/evaluation_eu.tex
\mysubsubsection{Experimental Setup}
We implemented \SPARTA in Go and used EGo to deploy trusted applications
.\footnote{EGo: \href{https://www.edgeless.systems/products/ego}{\nolinkurl{edgeless.systems/products/ego}}; Intel SGX: \href{https://sgx101.gitbook.io/sgx101}{\nolinkurl{sgx101.gitbook.io/sgx101}}. Accessed: 2026-03-19.}
We implemented the \RequestMan as a Redis queue, used \textit{govaluate} to evaluate \FEEL expressions and access policies, and \textit{Expr} to evaluate aggregation functions.\footnote{Redis queue: \href{https://github.com/rq/rq}{\nolinkurl{github.com/rq/rq}}; govaluate: \href{https://github.com/Knetic/govaluate}{\nolinkurl{github.com/Knetic/govaluate}}; expr: \href{https://github.com/expr-lang/expr}{\nolinkurl{github.com/expr-lang/expr}}. Accessed: 2026-03-19.}
We ran all experiments on a \TEE hosted by an Intel Xeon Gold 5415+ CPU with \qty{128}{\giga\byte} of RAM, Ubuntu 22.04, and Intel SGX enabled with a PRMRR 
size of \qty{2}{\giga\byte}. Each experiment was repeated \num{10} times, and we report mean values.
Our implementation's source code, together with input data and outputs, is available at \href{https://github.com/apwbs/SPARTA}{\nolinkurl{github.com/apwbs/SPARTA}}.

Four questions drive our evaluation:
\begin{questiilist}
	\item\label{qst:scala}%
    Does \SPARTA scale well with input size? 
	\item\label{qst:overhead}%
    How much overhead encryption and \TEE execution introduce compared to 
    an unprotected baseline running on the same hardware?
	\item\label{qst:xdomain}%
    Can {\SPARTA} operate with 
    decision models across different domains?
    \item\label{qst:correct}
    How effectively does the dual-storage system reduce memory consumption?
\end{questiilist}

To answer these questions and test the distinct features of our system, we resort to the following three datasets. 
\begin{inparadesc}
    \item[\tMED~(real-world medical dataset)] is the  Decision Management Community challenge's dataset (\textit{Determine Medical Service Coverage}, released in Feb.\ 2021). It contains \num{16000} healthcare-related records. We use it for scalability and overhead measurement because of its large size.
    \item[\tCARD (real-world financial dataset)] is the \textit{Card Approval Decision} challenge's dataset (Oct.\ 2021). Through it, we evaluate cross-domain performance with diversified and complex predicates.
    \item[\tVAX~(synthetic vaccine dataset)] is a synthetic vaccine-distribution dataset we built based on our running example. 
\end{inparadesc}
In response to \ref{qst:scala}, we use the \tMED dataset to analyze scalability based on the number of input records, columns, and rules. We address \ref{qst:overhead} through a controlled comparison measuring encryption and {\TEE} overhead against a non-encrypted baseline running on the same machine. We demonstrate cross-domain performance (\ref{qst:xdomain}) with the \tCARD dataset. We use the \tVAX dataset to answer \ref{qst:correct} 
measuring the dual-storage memory optimization. 

\mysubsubsection{Results}
\mysubsubsection{\tMED} 
The original dataset contains \num{16369} {rules}, with \num{7} {input} and \num{3} {output} columns. To increase the number of columns, we artificially replicated the input columns to reach a total of \num{28}. Although artificial, this setting allows us to evaluate system performance under higher-dimensional input across multiple parameters (\ie rules, columns, and records). 
Unless specified, all experiments below used \num{300} 
randomly selected rules from the original benchmark. 
\mypara{Scalability}
\Cref{fig:tests:execution-analysis} shows execution time while increasing the number of records and columns. We vary the number of records from \num{1000} to \num{16000} with \qty{7}, \qty{14}, and \num{28} columns, and then vary the number of columns from \num{1} to \num{28} with \num{1000}, \num{8000}, and \num{16000} records. In both cases, the growth is linear. \Cref{fig:tests:increasing-rules} shows the same trend when increasing the number of rules from \num{100} to \num{1500} with \num{16000} records and \num{28} columns, representing the most demanding configuration tested.

	\begin{figure*}[tb]
		\subfloat[With \num{7}, \num{14}, and \num{28} columns]{{\includegraphics[width=0.49\columnwidth]{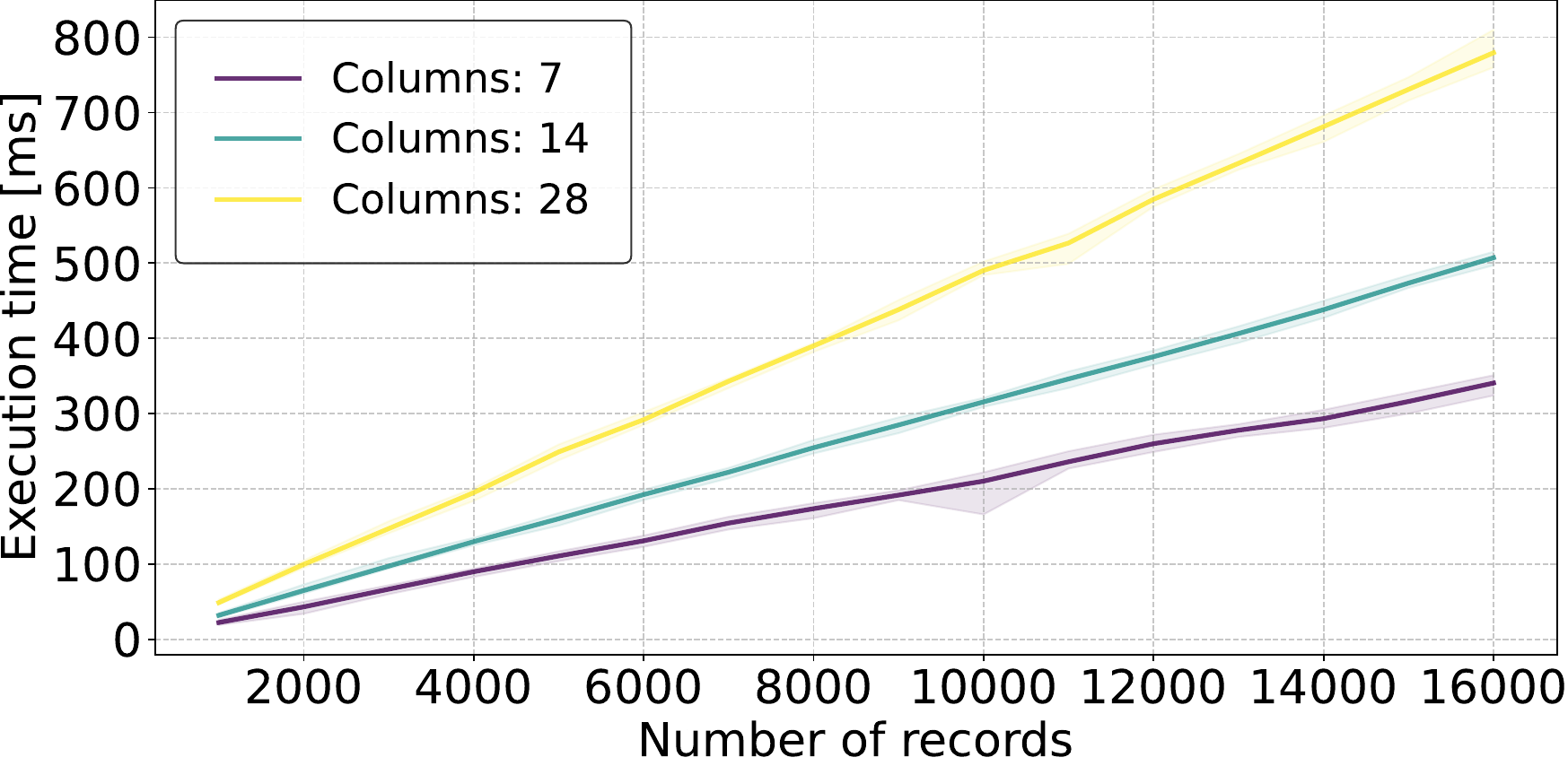} }  
			\label{fig:tests:fixed-columns}}
		\subfloat[With \qty{1}{k}, \qty{8}{k}, and \qty{16}{k} records]{{\includegraphics[width=0.49\columnwidth]{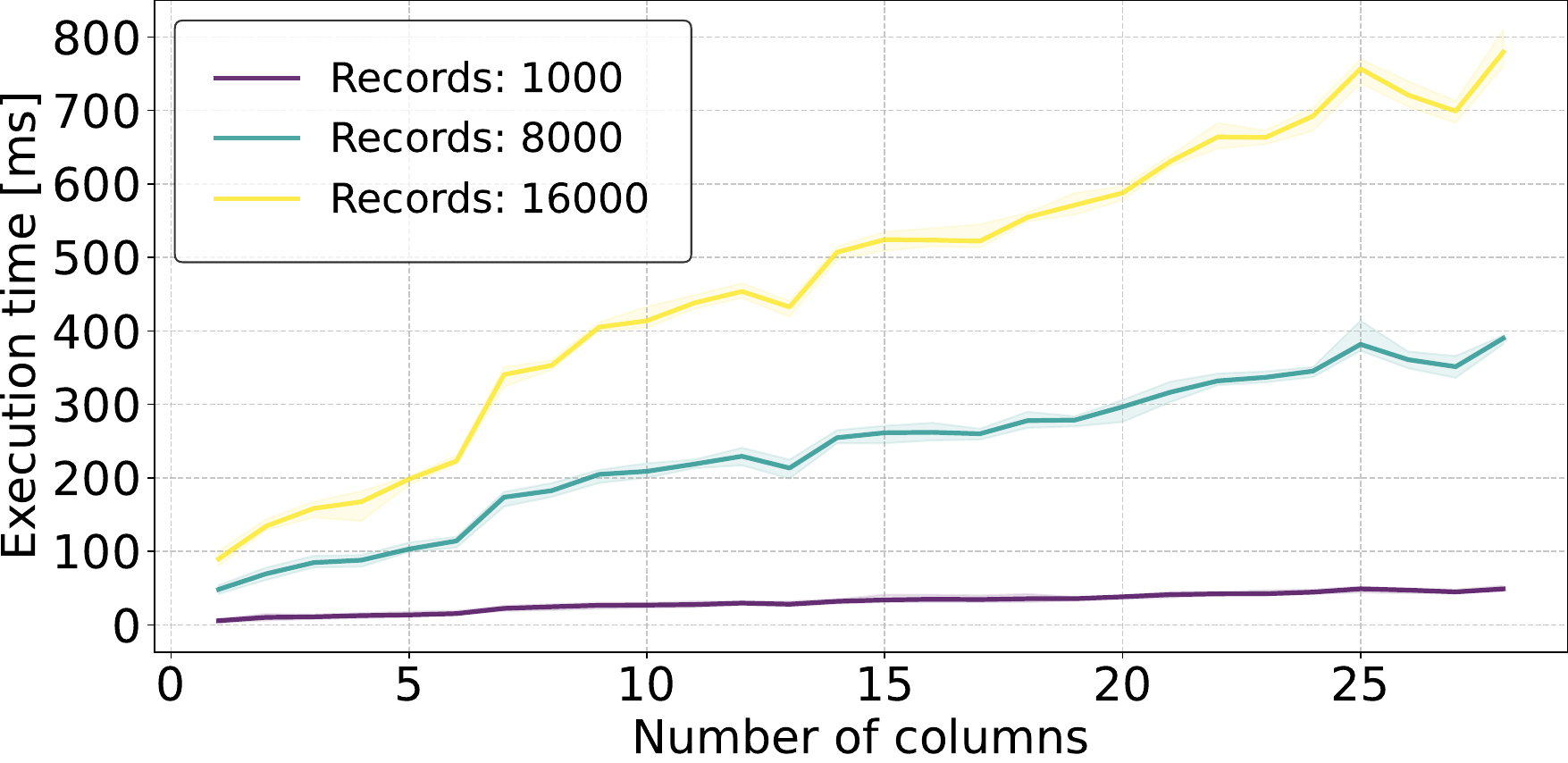} }
			\label{fig:tests:fixed-users}}
		\caption{Decision execution time vs.\ number of records (a) and columns (b)}
		\label{fig:tests:execution-analysis}
	\end{figure*}

\mypara{Aggregation overhead}
To measure the cost of aggregation, we augment the input with up to \num{21} columns, each of which computes an aggregation function over different input fields.
In particular, \Cref{fig:tests:aggregation} shows that decryption and aggregation add extremely limited overhead to the decision evaluation execution time, and that this overhead grows almost linearly. 

    \begin{figure*}[tb]
		\begin{floatrow}
			\ffigbox{\includegraphics[width=\columnwidth]{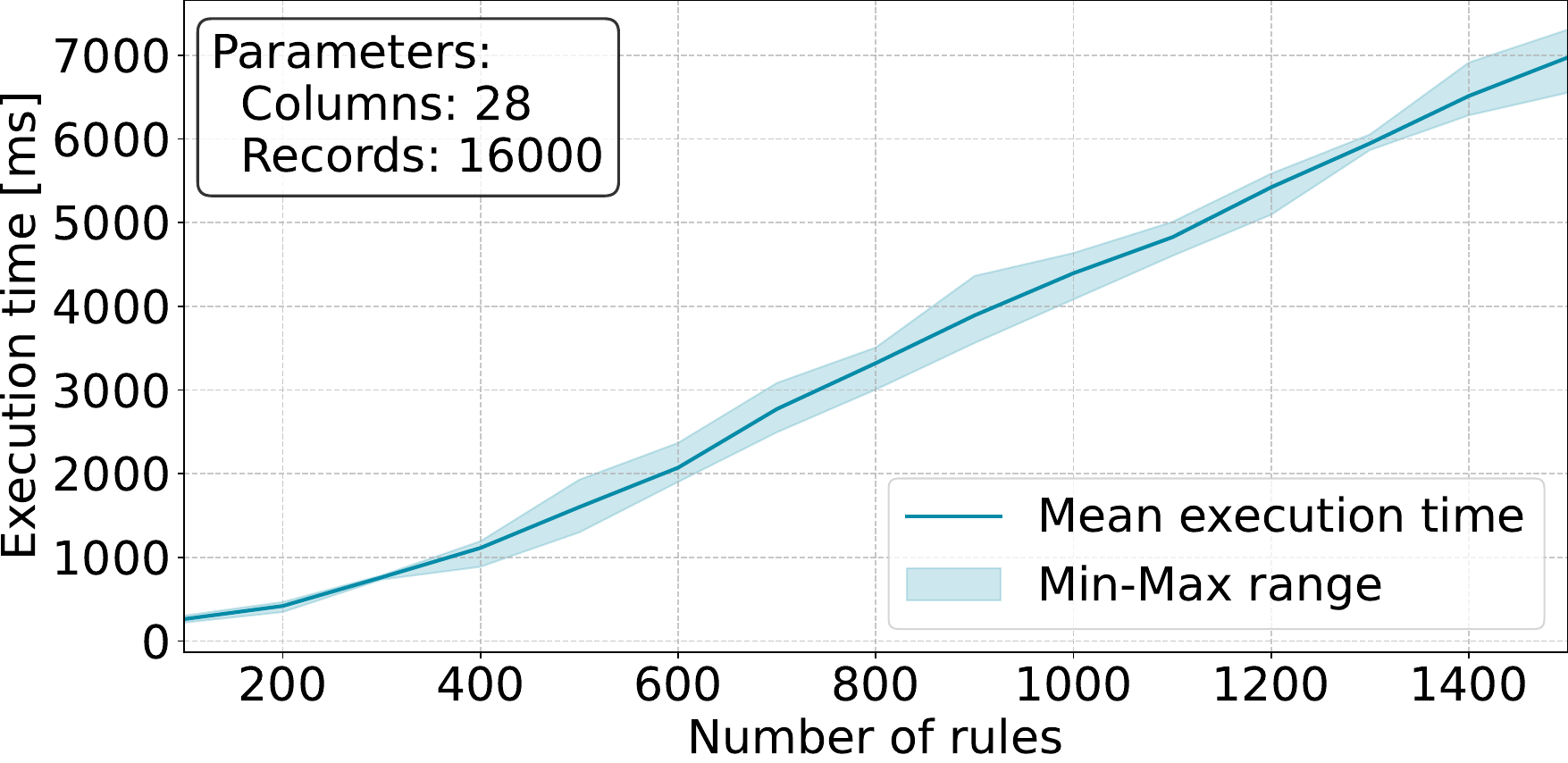}}{\caption{Runtime 
            vs.\ number of rules}\label{fig:tests:increasing-rules}}
			\ffigbox{\includegraphics[width=\columnwidth]{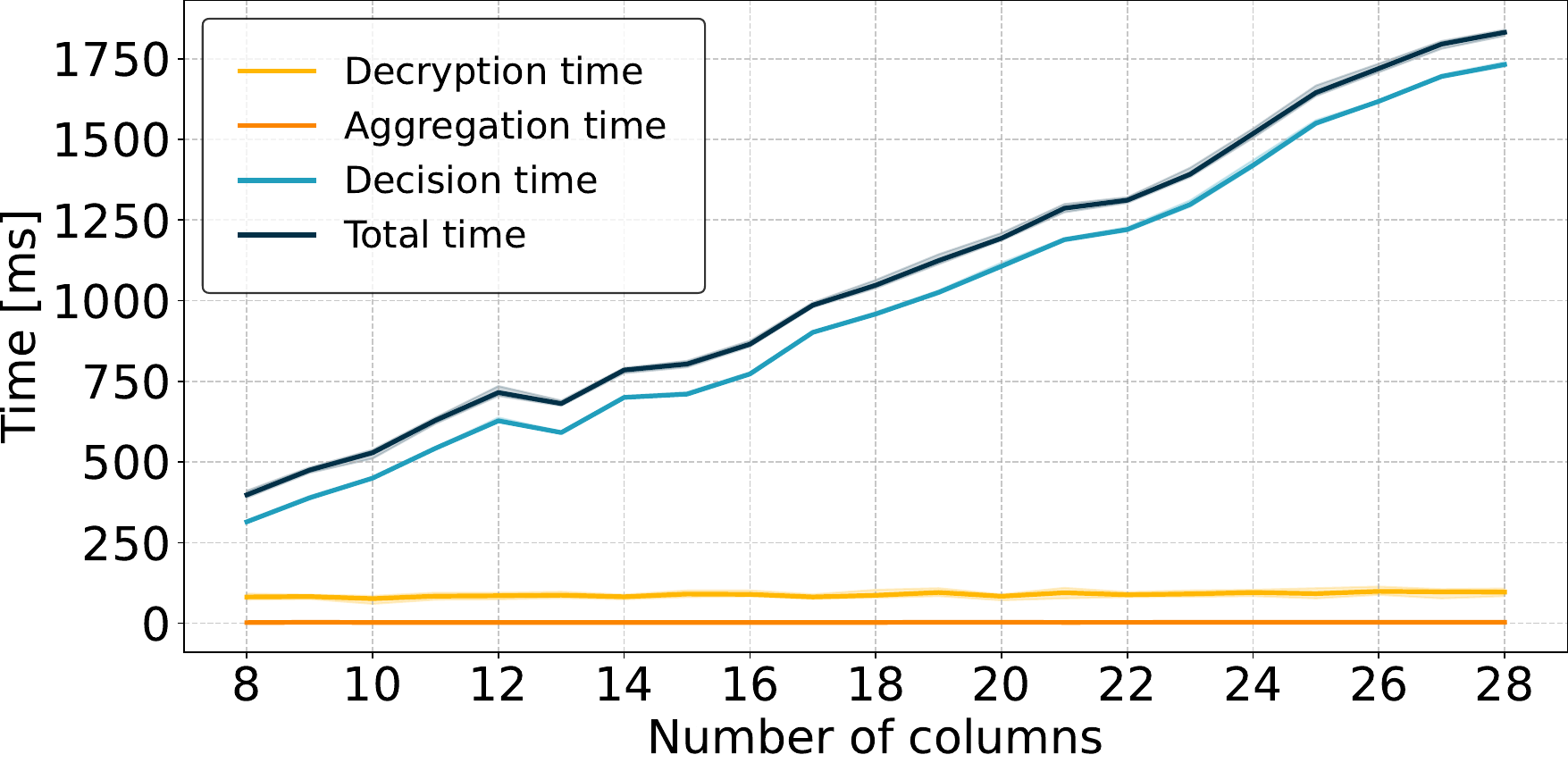}}{\caption{Runtime with prior aggregation}\label{fig:tests:aggregation}}
		\end{floatrow}
	\end{figure*}

\mypara{Memory usage}
\Cref{fig:tests:tee-ram-aggregation-comparison} reports enclave memory usage with and without the aggregation step. Each run begins with the \CCU's initialization and ends when the decision process terminates.
Both images show a reasonable memory footprint throughout the computation, with an expected peak during decryption.

    \begin{figure*}[tb]
		\subfloat[Without prior aggregation]{{\includegraphics[width=0.49\columnwidth]{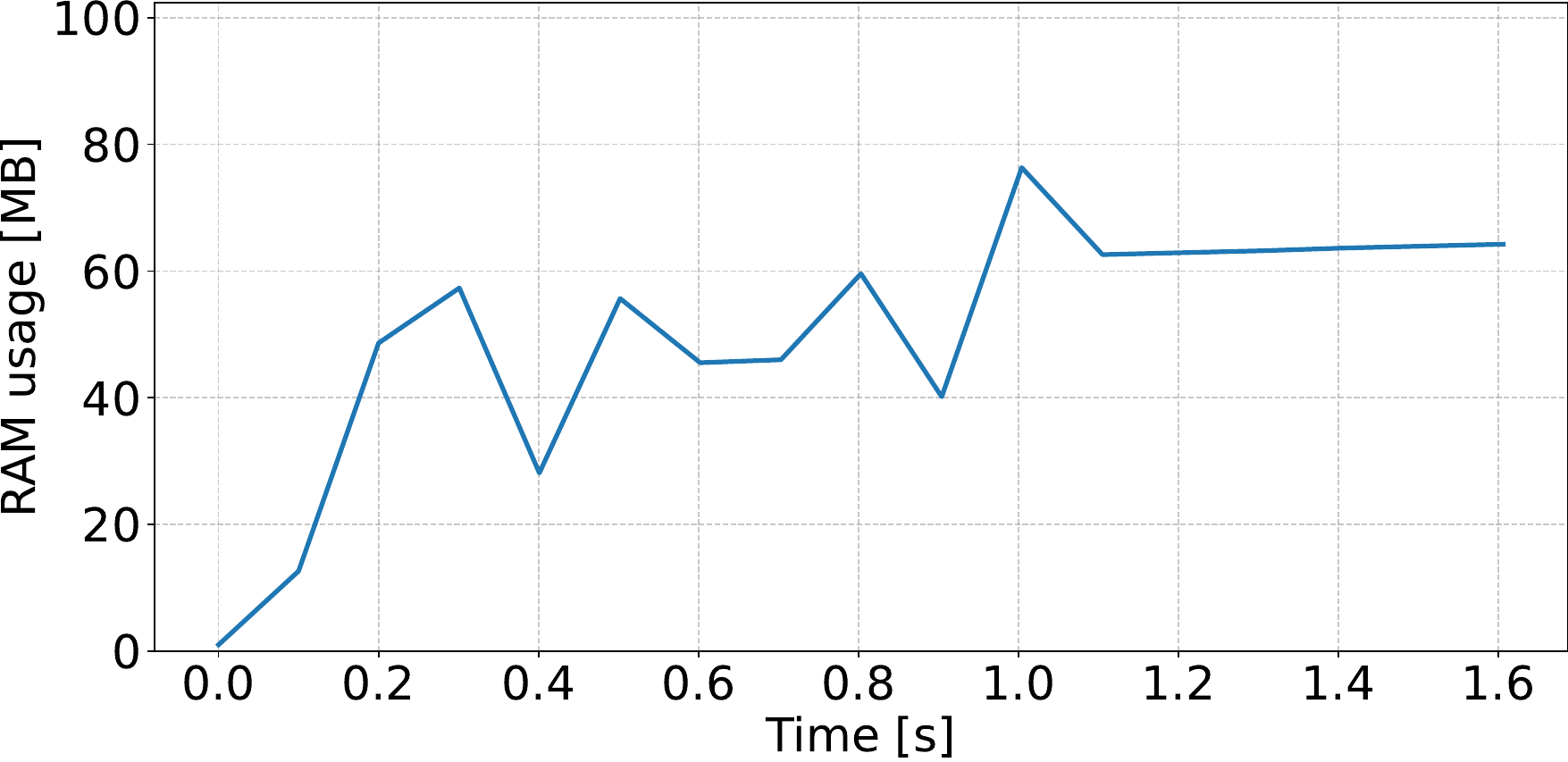} }
			\label{fig:tests:tee-ram-no-aggregation}}
		\subfloat[With prior aggregation]{{\includegraphics[width=0.49\columnwidth]{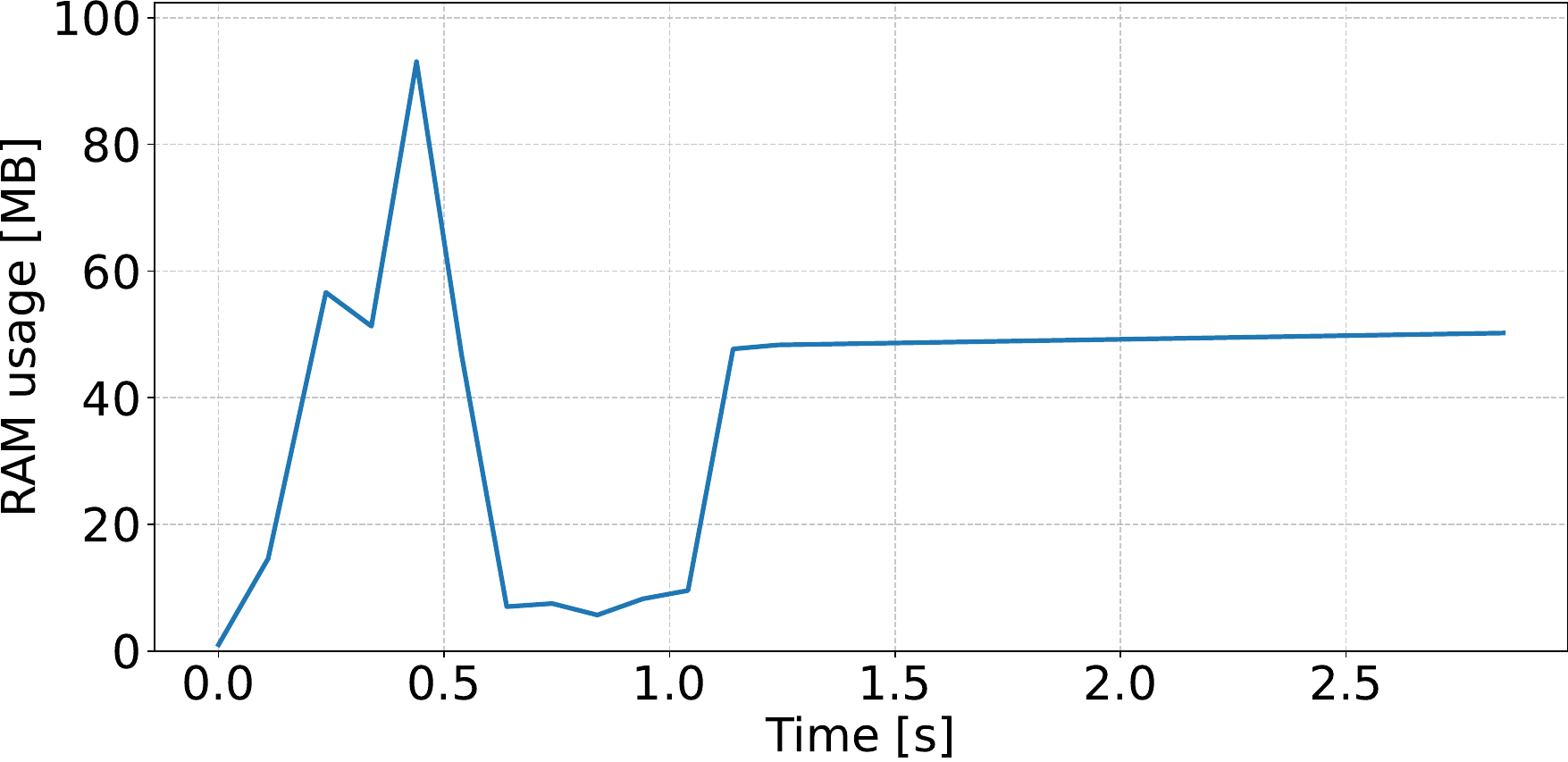} }
			\label{fig:tests:tee-ram-aggregation}}
		\caption{Memory usage over time} 
		\label{fig:tests:tee-ram-aggregation-comparison}
	\end{figure*}


The previous three experiments address \ref{qst:scala} by demonstrating linear scalability with respect to input size, number of columns, and rule complexity, and also bound memory usage within the \TEE.
Next, we address \ref{qst:overhead} by quantifying the overhead introduced by encryption and TEE execution.

%
\mypara{Encryption overhead}
We compare a light configuration in which all records share a single symmetric key with the actual heavy configuration, in which each record is encrypted with a different key (see \nWPii stage). Note that the light configuration is used only as a baseline; from a security perspective, compromising that key would expose the entire dataset, whereas in the heavy one, a key compromise affects at most one record. \Cref{fig:tests:encryption-comparison} shows that the heavy configuration adds a reasonable overhead of at most \qty{80}{\milli\second} across all tested configurations, while providing higher data confidentiality and resilience against key compromise.

\begin{figure*}[tb]
	\subfloat[Records for \num{28} fixed
    columns]{{\includegraphics[width=0.49\columnwidth]{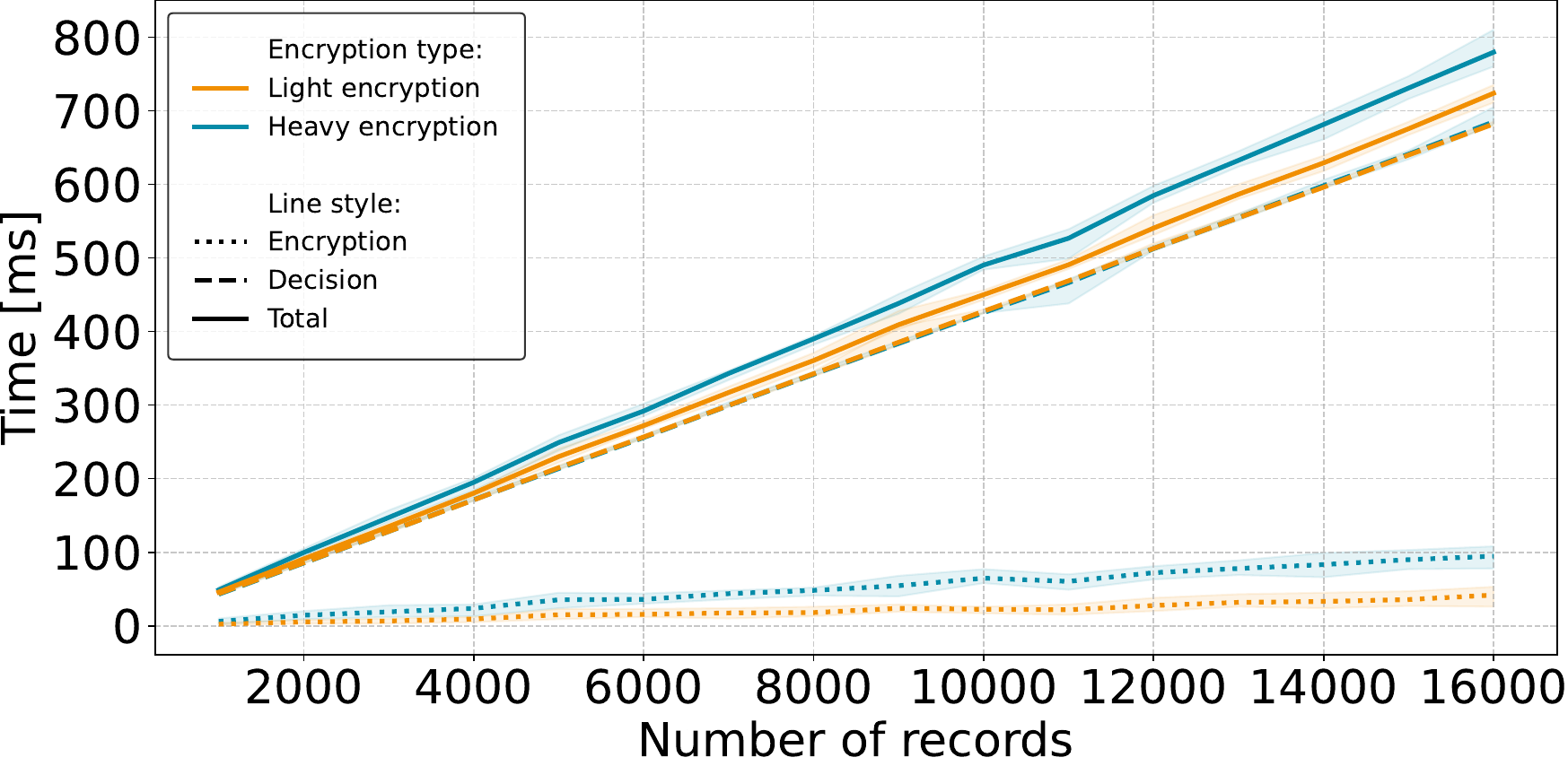} }
		\label{fig:tests:encryption-comparison-28-columns}}
	\subfloat[Columns for \qty{16}{k} fixed
    users]{{\includegraphics[width=0.49\columnwidth]{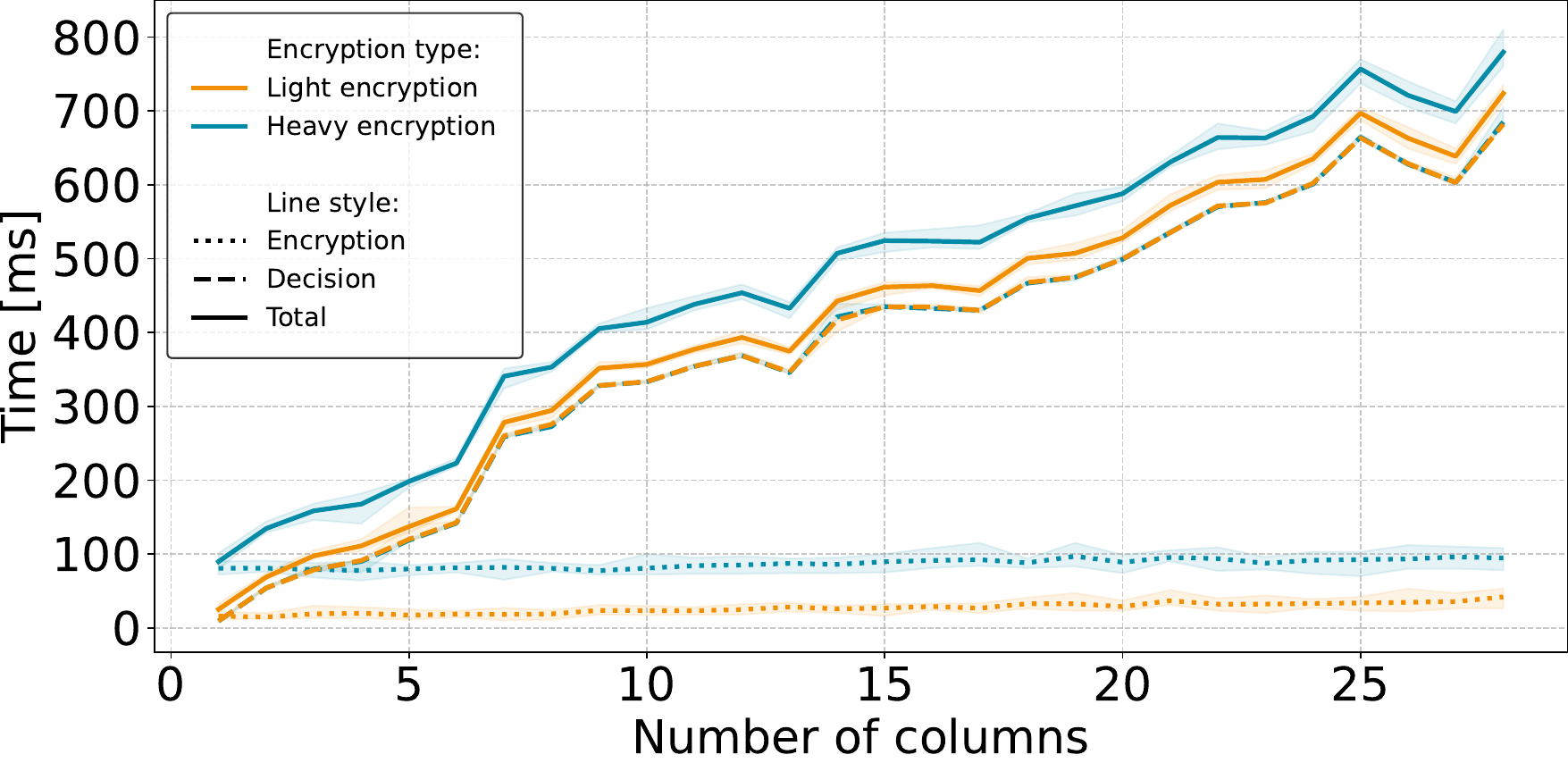} }  
		\label{fig:tests:encryption-comparison-16-users}}
	\caption{Comparison between light and heavy encryption}
	\label{fig:tests:encryption-comparison}
\end{figure*}

\mypara{TEE vs.\ server}
We compare \SPARTA's performance against an unprotected baseline that executes the same logic outside a \TEE (on the same machine) and without encryption.
\Cref{fig:tests:server-vs-tee} shows that the additional cost of fully protected execution is limited and mainly attributable to the decryption phase and enclave execution rather than to the decision logic itself. The gap between the two configurations remains stable as the number of records and columns increases, indicating that the overhead does not compound with input size.

\begin{figure*}[tb]
	\subfloat[Comparison for \num{28} columns]{{\includegraphics[width=0.47\columnwidth]{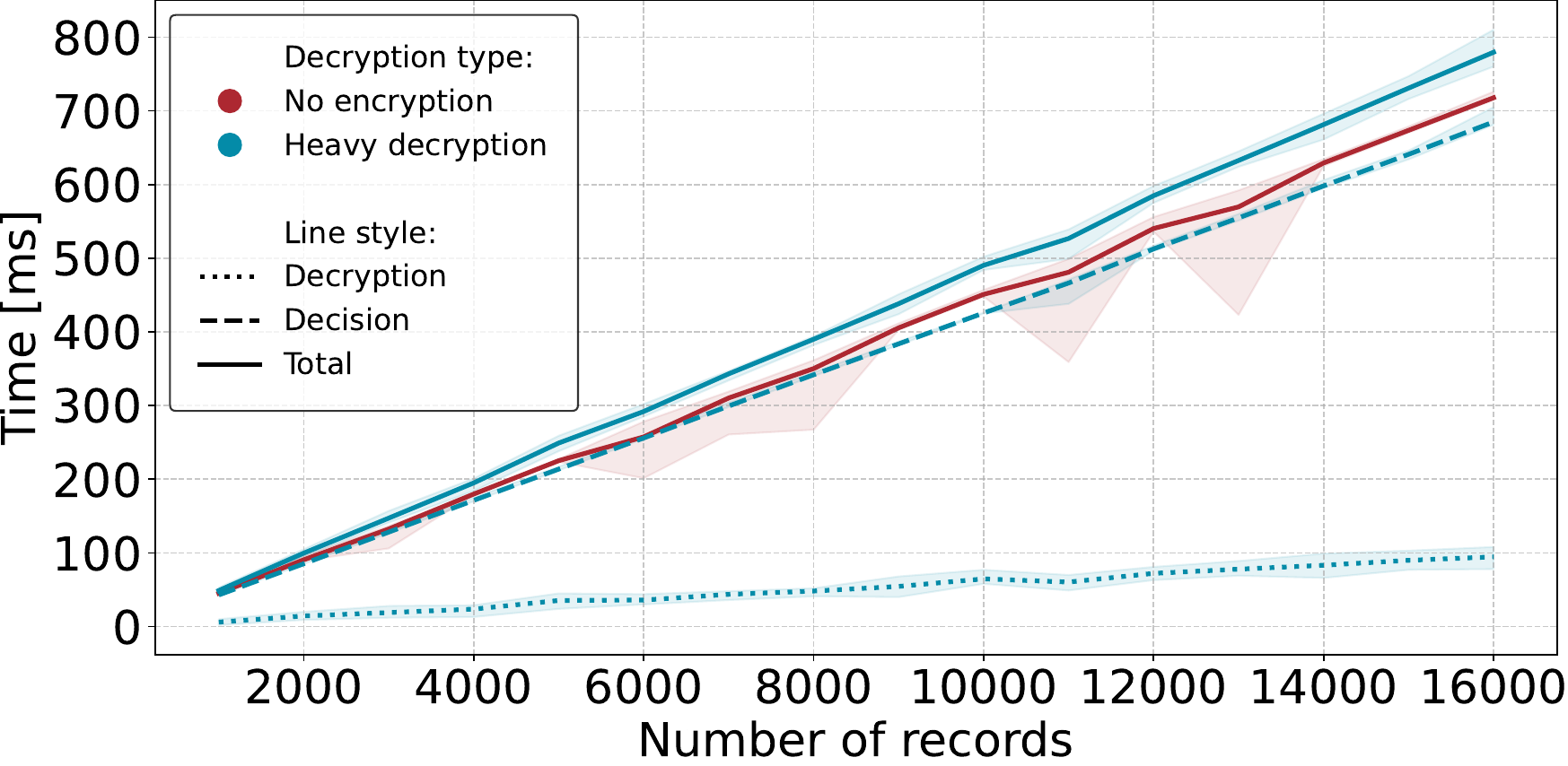} }  
		\label{fig:tests:server-vs-tee-columns}}
	\subfloat[Comparison for \qty{16}{k} records]{{\includegraphics[width=0.47\columnwidth]{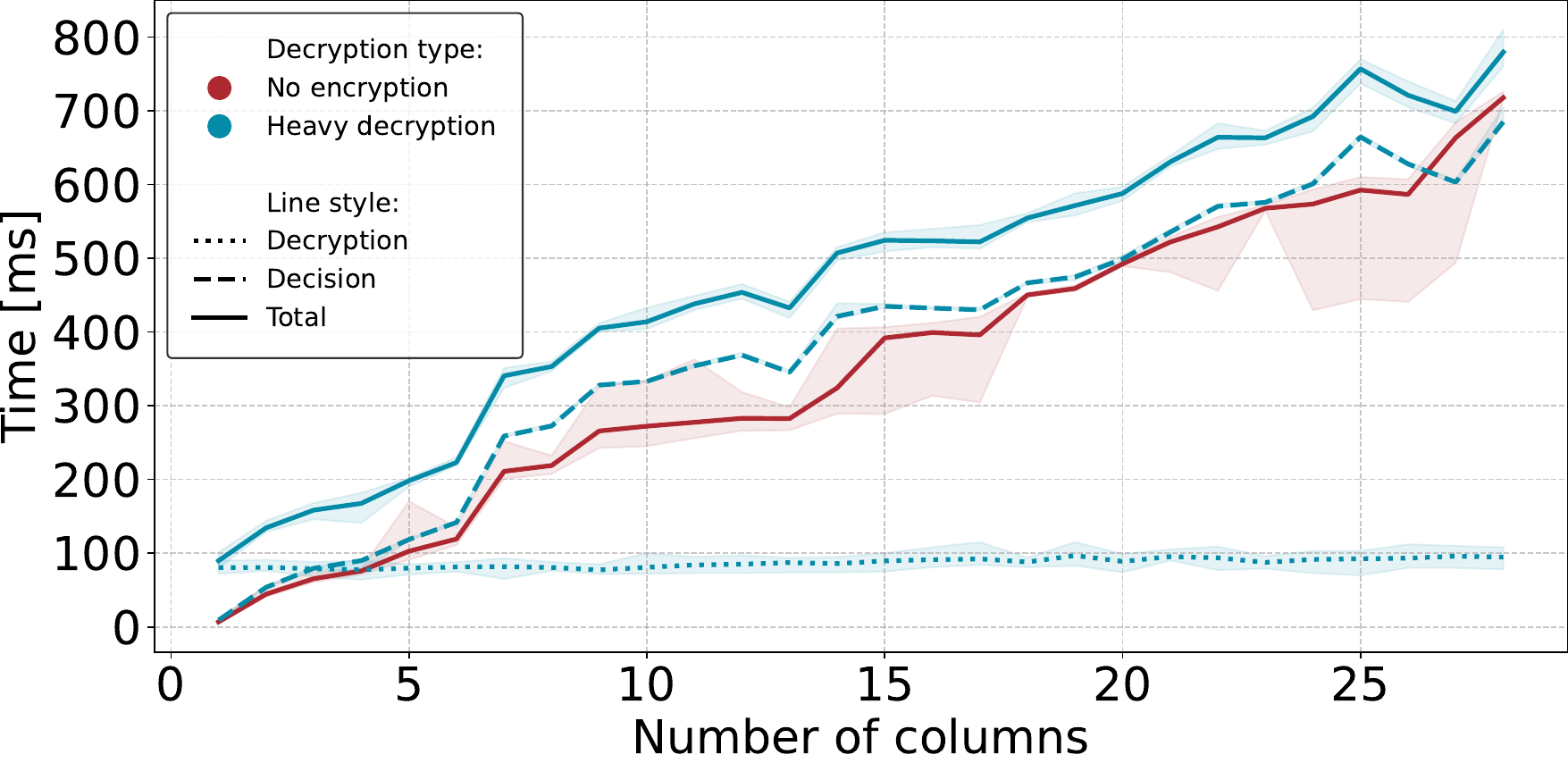} }
		\label{fig:tests:server-vs-tee-users}}
	\caption{Execution time comparison between server and \TEE}
	\label{fig:tests:server-vs-tee}
\end{figure*}


\mysubsubsection{\tCARD}
This benchmark contains \num{3} input columns and \num{13} rules with relational and interval-based conditions from the financial domain (to cater to \ref{qst:xdomain}), featuring more complex predicates than those in \tMED and \tVAX.
Since the original challenge does not provide input data, we generated inputs covering all rules and duplicated them up to \num{100000} records.
%
%
\Cref{fig:tests:barplot} shows the execution times across the three datasets.
\tMED requires the highest decision time due to its larger rule set. In contrast, the decryption time remains consistent across the datasets, confirming \SPARTA's ability to maintain stable performance across domains with varying logic complexity and input scale.

\mysubsubsection{\tVAX}
We built the dataset by generating realistic input examples for a nationwide vaccine campaign for patients, vaccination centers, and carriers, based on the decision logic in \gls{dmn} (see \cref{tab:dmn}). 
We use this synthetic scenario to evaluate the dual-storage strategy introduced in the workflow (see the \nWP stage in \cref{sec:approach:execution}) as an answer to \ref{qst:correct}.
Each patient provides full records, but only a subset of the fields is required for decision making: \num{6} out of \num{38} fields for patients, \num{8} out of \num{50} for carriers, and \num{4} out of \num{35} for vaccination centers.
To enhance privacy and efficiency, we store the lightweight version containing only the decision-relevant fields, alongside the original one for traceability. Notice that this dual-storage strategy allows the \DOwner to submit data without knowledge of the decision-critical fields, preventing inference attacks and further preserving logic confidentiality.
\Cref{fig:tests:re:memory-saving} displays a memory saving in the range of \SIrange{77.2}{83.8}{\percent}.


\begin{figure*}[tb]
    \begin{floatrow}
        \ffigbox{\includegraphics[width=.92\columnwidth]{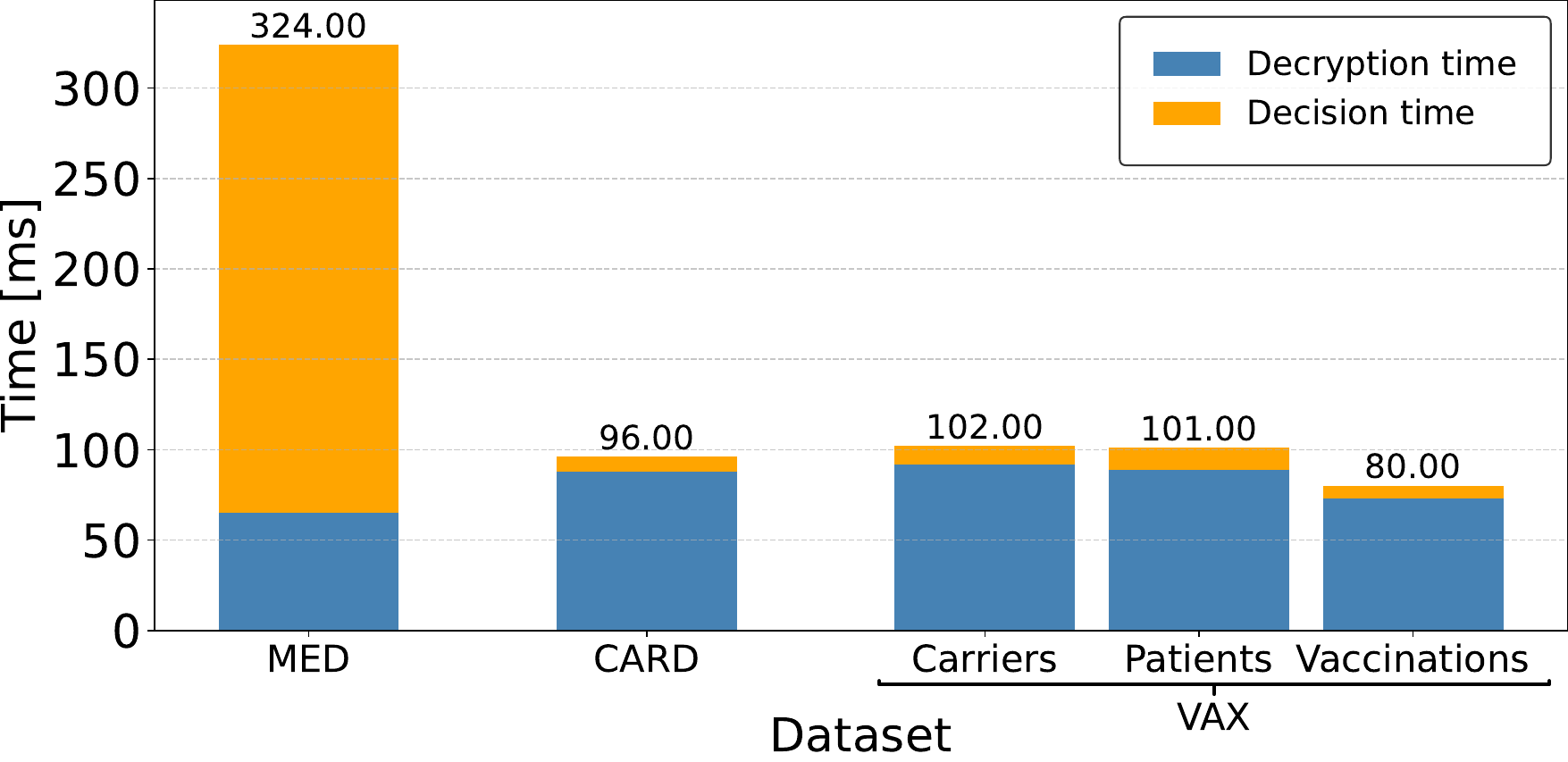}}{\caption{Execution time with \tVAX, \tMED, and \tCARD datasets}\label{fig:tests:barplot}}
        \ffigbox{\includegraphics[width=.92\columnwidth]{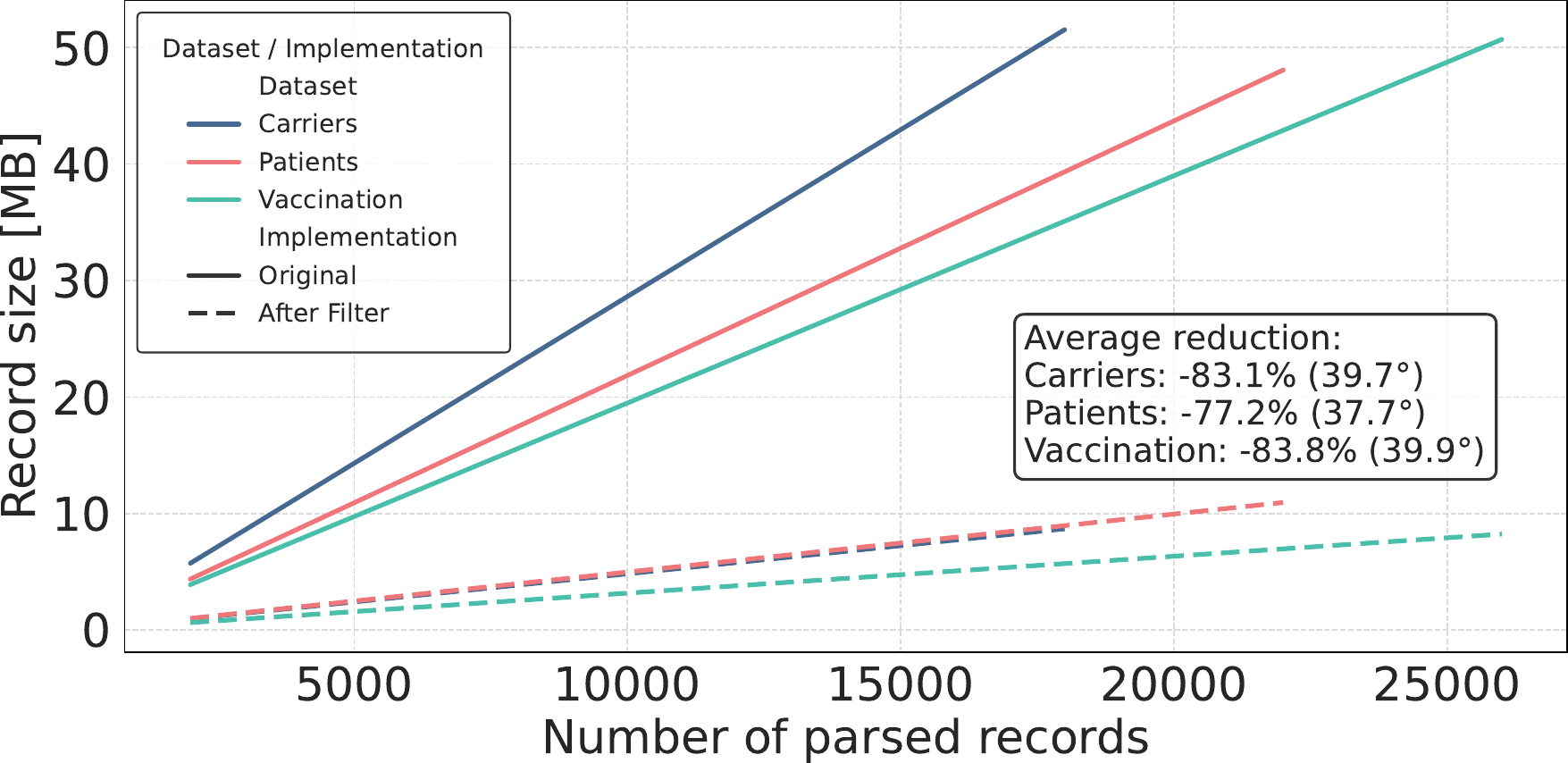}}{\caption[Memory savings]{Memory consumption with the dual-storage strategy}\label{fig:tests:re:memory-saving}}
    \end{floatrow}
\end{figure*}


Finally, to assess scalability beyond the original benchmark size, we conducted an additional stress test using the \tVAX dataset with the decision logic of \cref{tab:dmn} and up to \num{1000000} patient records as input. \SPARTA completed decision execution in approximately \qty{0.5}{\second}, confirming that the linear scaling trend extends well beyond the original dataset.

%% file: sections/sota.tex
In recent years, numerous approaches have been proposed in the field of decision-making support systems.
Haarmann et al.~\cite{haarmann2018dmn} propose executing DMN decision models on the Ethereum blockchain. However, blockchain-based approaches inherently rely on public ledgers, making them unsuitable for confidential data. In contrast, \SPARTA enables secure, private decision support while maintaining trust, transparency, and auditability, with access policies that ensure only authorized participants can interact with the system.
Zhang et al.~\cite{Zhang-privacydecisions} introduce a privacy-preserving decision-making system on a blockchain for secure voting, proposing a distributed batch key-generation protocol and a two-stage voting scheme. While this addresses privacy issues on blockchain, it is tailored to voting only, whereas \SPARTA supports user-definable decision logic.
Other solutions~\cite{Rahulamathavan} enable privacy-preserving clinical decisions using homomorphic encryption over Gaussian kernel-based SVMs. While effective for automated classification, these approaches incur significant performance overhead and fix the decision logic. In contrast, \SPARTA uses \AESGCM encryption, which introduces low overhead and allows users to specify custom decision logic.
Liu et al.~\cite{LIU2018825} present HPCS, a hybrid privacy-preserving clinical decision support system for health monitoring and disease prediction, in which a cloud server runs neural network predictions using the Paillier Cryptosystem with Threshold Decryption. Compared to \SPARTA, HPCS incurs higher overhead and limits decision logic to prediction-based decisions, whereas we support user-defined logic.
Xie et al.~\cite{TEBDS} present TEBDS, a TEE-based blockchain data-sharing system for IoT that relies on a centralized key management center (KMC). In contrast, \SPARTA derives all keys within the TEE from a shared seed, 
removing this single point of failure. Moreover, \SPARTA uses a public blockchain only to store IPNS names pointing to IPFS content-addressed data, delegating all operations to the TEE. 

%% file: sections/conclusion.tex
We presented \SPARTA, an approach that enables automated decision-support over confidential data in multi-party settings while preserving data secrecy. \SPARTA combines \TEE-based execution with per-record key derivation, distributed storage, and blockchain notarization to solve the confidential verifiable decision-support problem, enabling multi-party decision-making that is confidential, verifiable, customizable, and practical. 
Our experiments on real-world and synthetic benchmarks confirm linear scalability and an overhead of less than \qty{80}{\milli\second} compared to unprotected execution.

Our solution presents some limitations. \SPARTA supports \CCU replication but does not yet implement load balancing, which would require non-trivial changes to session management. Our prototype targets Intel SGX; porting to other \TEE platforms such as AMD SEV, Arm CCA, or Intel TDX is a natural next step. Our threat model also excludes side-channel attacks~\cite{Cauldron}, though existing mitigation techniques~\cite{Obelix} and deployment within protected environments can reduce this risk. Our implementation covers only standard aggregation functions (\SmallCode{avg}, \SmallCode{sum}, \SmallCode{max}, \SmallCode{min}) and single-operation filters over data. 
Integrating \TEE-based smart contracts~\cite{Ekiden} could further broaden the range of automated decision logic that the system supports.
Finally, although our experiments with public benchmarks show good performance and scalability with low overhead, on-field use of \SPARTA in real-world settings is paramount for validation.

%% file: sections/formalAppendix.tex
We first introduce the required notation and cryptographic tools.
\begin{asparadesc}
	\item[Notation.] Let $\NN$ be the set of all natural numbers. We denote the security parameter by $\lambda \in \NN$.
	Every algorithm takes the security parameter $\lambda$ as input~(in unary, \ie $1^\lambda$). When an algorithm has multiple inputs, $1^\lambda$ is typically omitted.
	A function is called negligible if, for every positive polynomial and all sufficiently large security parameter $\lambda$, the function grows at most as the inverse of the polynomial in $\lambda$.
	
	\item[Random Oracle.] In the Random Oracle Model (ROM), all the parties have access to a truly random hash function. In particular, when a value $v$ is given as an input from a party to the Random Oracle~(RO), the latter samples a random answer $r$, stores the pair~$(v,r)$, and outputs $r$ to the party. If the RO is queried on the same value $v$ multiple times, the same answer $r$ is output.
	It is heuristically shown that cryptographic hash functions such as \SHA-256 behave as a RO.
	
	\item[Authenticated Encryption.] A Symmetric Key Encryption scheme with key space $\keyspace$ is composed of the following algorithms:
	\begin{iiilist}
		\item $\Gen(1^\lambda)$: The randomized key generation algorithm outputs a key $k$ in the key space $\keyspace$.
		\item $\Enc(k,m)$: The randomized encryption algorithm, on input the key $k$ and a message $m$, outputs a ciphertext $c$.
		\item $\Dec(k,c)$: The decryption algorithm, on input the key $k$ and a ciphertext $c$, outputs a message $m$, or $\bot$ if decryption fails.%
	\end{iiilist}
	An authenticated encryption scheme should satisfy the following properties:
	\begin{iiilist}
		\item \textbf{Correctness:} For each message $m$ and each $k\in\keyspace$, $\Pr[\Dec(k,\Enc(k,m))=m]=1$.
		\item \textbf{Unforgeability:} The probability that a PPT adversary $\adv$ wins the following experiment~(\ie outputs 1) is negligible:
		\begin{inparaenum}[\itshape(a)\upshape]
			\item Obtain $k$ from $\Gen(1^\lambda)$.
			\item The adversary, after making at most polynomially many queries to the encryption oracle $\Enc(k,\cdot)$, outputs a ciphertext $c$.
			\item Let $m=\Dec(k,c)$. Output 1 if $m\ne\bot$ and $c$ was not queried to the encryption oracle by $\adv$. Output 0 otherwise.
		\end{inparaenum}
		\item \textbf{CCA-Security:} The probability that a PPT adversary $\adv$ wins the following experiment~(\ie outputs 1) is negligibly close to $1/2$:
		\begin{inparaenum}[\itshape(a)\upshape]
			\item Obtain $k$ from $\Gen(1^\lambda)$.
			\item The adversary, with oracle access to $\Enc(k,\cdot)$ and $\Dec(k,\cdot)$, outputs messages $m_0,m_1$.
			\item Choose a uniform bit $b\in\{0,1\}$, compute ciphertext $c=\Enc(k,m_b)$, and send $c$ to $\adv$.
			\item The adversary, with oracle access to $\Enc(k,\cdot)$ and restricted oracle access to $\Dec(k,\cdot)$ (excluding $c$), outputs a bit $b'$.
			\item The experiment outputs $1$ if $b'=b$, else $0$.
		\end{inparaenum}
	\end{iiilist} 
\end{asparadesc}


\bigskip\noindent
We can now state the formal properties we require from our protocol.

\begin{asparadesc}
    \item[Correctness.] For any function $f$ and any set of messages $M$ such that $f(M)=y$, if $M$ is provided to the~\CCU for function $f$ by the Data Provider obtaining ciphertexts $C$, and subsequently, ciphertexts $C$ are provided by the Decider back to the~\CCU obtaining output $y'$, then $y'=f(M)=y$.
    
    \item[Data Privacy.] A PPT adversary $\adv$ wins the following experiment with probability negligibly close to $1/2$:
	\begin{iiilist}
	    \item $\adv$ chooses a function $f$. $\adv$ has oracle access to the~\CCU acting as Data Provider (sending arbitrary sets of messages $M$ and receiving corresponding ciphertexts $C$) and oracle access to the~\CCU acting as Decider (sending arbitrary sets of ciphertexts $C$ and receiving back $f(M)$ if $M$ correctly decrypts from $C$). Then, $\adv$ chooses two sets of messages $M_0,M_1$ such that $f(M'_0)=f(M'_1)$ where $M'_b$ can be any combination of messages in $M_b$.
	    \item Choose a random bit $b$, produce ciphertexts $C_b$ of messages $M_b$, and send $C_b$ to $\adv$.
	    \item $\adv$ has oracle access to the~\CCU as a Data Provider to send arbitrary sets of messages $M$, and obtain the corresponding ciphertexts $C$ back. It has further oracle access to the~\CCU as a Decider by sending arbitrary ciphertexts $C$ and obtaining back $f(M)$ if $M$ correctly decrypts $C$. Finally, outputs a bit $b'$. 
	    \item If $b'=b$, output 1; otherwise, output 0.
	\end{iiilist}
	
	\item[Data Integrity.] A PPT adversary $\adv$ wins the following experiment with negligible probability:
	\begin{iiilist}
	    \item $\adv$, acting as Data Provider, has oracle access to the~\CCU submitting arbitrary message sets $M$ and receiving ciphertexts $C$. $\adv$ outputs ciphertexts $C^*$.
	    \item Decrypting ciphertexts $C^*$ yields messages $M^*$. Output 1 if no message in $M^*$ equals $\bot$ and no subset of $M^*$ was previously queried; else, output 0.
	\end{iiilist}
	
	\item[Forward/Backward Secrecy/Integrity.] If all ciphertext keys at time $t$ are leaked, ciphertexts encrypted at any different time $t'$ remain secure, and their underlying messages cannot be decrypted or forged.
\end{asparadesc}


\subsubsection{The protocol.} Let $(\Gen,\Enc,\Dec)$ be an authenticated encryption scheme with key space $\keyspace$ and $H: \{0,1\}^* \rightarrow \keyspace$ be a cryptographic hash function modeled as a RO. 
The protocol is as follows:
\begin{inparaenum}[(\bfseries1)]
    \item \textbf{Setup:} The~\CCU is initialized with functions $\{f_j\}$ and a uniform seed $s\in\{0,1\}^\lambda$.
    \item \textbf{Data Storage:}
    \begin{iiilist}
        \item The Data Provider establishes a secure channel with the~\CCU, chooses message set $M$, and sends it to the~\CCU.
        \item The~\CCU generates a unique ID $T_i$, key $k_i=H(s,T_i)$, ciphertexts $c_i=\Enc(k_i,m_i)$ for $m_i\in M$, and stores ciphertexts publicly.
    \end{iiilist}
    \item \textbf{Decision:}
    \begin{iiilist}
        \item Decider establishes a secure channel with the~\CCU and sends ciphertext set $C$ and function index $j$.
        \item~\CCU takes the unique ID $T_i$, keys $k_i=H(s,T_i)$, decrypts ciphertexts to messages $m_i$, and outputs decision $f_j(\{m_i\})$.
    \end{iiilist}
\end{inparaenum}

\begin{sloppypar}
\begin{theorem}
    Assuming that $(\Gen,\Enc,\Dec)$ is an Authenticated Encryption scheme, $H$ is modeled as a Random Oracle and the $\CCU$ is a secure TEE, the protocol described above satisfies Correctness, Data Privacy, Data Integrity, and Forward/Backward Secrecy and Integrity.
\end{theorem}
\end{sloppypar}
\begin{proof}
    Recall that the~\CCU is modeled as a secure TEE. Hence, the fixed functions $\{f_j\}$ cannot be modified, internal states remain secret, and outputs are always correctly computed with respect to these functions.
    Correctness is ensured by the correctness property of the authenticated encryption scheme $(\Gen,\Enc,\Dec)$ and by the fact that the~\CCU's functions are fixed and unmodifiable.
    Data Privacy is guaranteed since the internal~\CCU state (i.e., the seed $s$) remains secret and the derived encryption keys $k_i=H(s,T_i)$ appear indistinguishable from random due to the RO assumption. Thus, Data Privacy reduces directly to the CCA-security of $(\Gen,\Enc,\Dec)$.
    Data Integrity follows from the Unforgeability property of the authenticated encryption scheme $(\Gen,\Enc,\Dec)$, combined with the trusted execution guarantees provided by the \CCU.
    Forward/Backward Secrecy and Integrity results directly from the RO-based key derivation. Given a leaked key $k=H(s,T)$, the one-way property of the RO prevents recovery of the seed $s$. Therefore, keys corresponding to ciphertexts encrypted at other timestamps $T'\ne T$ remain computationally infeasible to derive.
\end{proof}